\documentclass[a4paper,UKenglish]{article}

\usepackage{amssymb,amsmath,amsthm}
\usepackage{graphicx}
\usepackage{xspace}
\usepackage{wrapfig}
\usepackage{a4wide}

\usepackage{lineno}

\usepackage{tikz}
\usetikzlibrary{calc,intersections,through,backgrounds}
\usetikzlibrary{arrows, decorations.pathreplacing}

\newcommand{\mytitle}{Non-Monochromatic and Conflict-Free Coloring on Tree Spaces and Planar Network Spaces}

\title{\mytitle}

\newcommand{\reals}{\mathbb{R}\xspace}
\newcommand{\Reals}{\reals}
\newcommand{\integers}{\mathbb{N}\xspace}

\newcommand{\naturals}{\integers}

\newcommand{\etal}{\emph{et~al.}\xspace}
\newcommand{\graph}{\mathcal{N}\xspace}

\DeclareMathOperator{\cf}{cf}

\DeclareMathOperator{\nm}{nm}
\newcommand{\cov}{cov}

\newcommand{\mypara}[1]{\vspace{10pt} \noindent \textbf{#1}}

\newcommand{\floor}[1]{\left\lfloor #1 \right\rfloor}
\newcommand{\ceil}[1]{\left\lceil #1 \right\rceil}
\renewcommand{\leq}{\leqslant}
\renewcommand{\geq}{\geqslant}
\newcommand{\tree}{\mathcal{T}}

\newcommand{\Col}{c\xspace}
\newcommand{\col}{\Col}

\newcommand{\TreeSpace}{\mathcal{T}\xspace}
\newcommand{\treespace}{\TreeSpace}
\newcommand{\Objects}{\mathcal{A}\xspace}
\newcommand{\objects}{\Objects}
\newcommand{\Core}{\mathcal{C}\xspace}
\newcommand{\core}{\Core}

\newcommand{\CFCN}[2]{X_{\cf}^{\mathrm{#1},\mathrm{#2}}}

\newcommand{\CFCNtree}[1]{\CFCN{tree}{#1}}

\newcommand{\CFCNtreepath}{\CFCNtree{paths}}
\newcommand{\CFCNtreetree}{\CFCNtree{trees}}

\newcommand{\NMCN}[2]{X_{\nm}^{\mathrm{#1},\mathrm{#2}}}

\newcommand{\NMCNtree}[1]{\NMCN{tree}{#1}}

\newcommand{\NMCNtreetree}{\NMCNtree{trees}}

\theoremstyle{plain}
\newtheorem{theorem}{Theorem}[section]
\newtheorem{lemma}[theorem]{Lemma}

\newtheorem{observation}[theorem]{Observation}

\definecolor{nicegreen}{rgb}{0,0.7,0.3}
\definecolor{niceyellow}{rgb}{0.9,0.7,0.07}
\definecolor{nicepurple}{rgb}{0.5,0.2,0.8}

\graphicspath{{./figures/}}
\author{Boris Aronov\thanks{New York University, USA. %
		BA has been partially supported by NSF Grants CCF-11-17336, CCF-12-18791, and CCF-15-40656, and by BSF grant 2014/170.}
        \and 
        Mark de Berg\thanks{TU Eindhoven, the Netherlands.
        MdB and AM are supported by the Netherlands' Organisation for Scientific
        Research (NWO) under project no.~024.002.003.}
        \and
        Aleksandar Markovic$^{\dagger}$
        \and
        Gerhard Woeginger\thanks{RWTH Aachen, Germany.}
        }

\begin{document}

\maketitle


\begin{abstract}
It is well known that any set of $n$ intervals in $\Reals^1$ admits a
non-monochromatic coloring with two colors and a conflict-free coloring
with three colors. We investigate generalizations of this result to
colorings of objects in more complex 1-dimensional spaces, namely
so-called tree spaces and planar network spaces.
\end{abstract}


\section{Introduction}
\emph{Conflict-free colorings}, or CF-colorings for short, were introduced by
Even~\etal~\cite{even-cf-03} and Smorodinsky~\cite{thesis-smorodinsky}
to model frequency assignment to base stations in wireless networks.
In the basic setting one is given a set~$S$ of objects in the
plane---often disks are considered---and the goal is to assign a
color to each object such
that the following holds: for any point~$p$ in the plane such that the set~$S_p:=\{D\in S\mid p\in D\}$
of objects containing~$p$ is non-empty, $S_p$ must contain an object whose color is different
from the colors of the other objects in~$S_p$. Even~\etal~proved, among other things,
that any set of disks admits a CF-coloring with~$O(\log n)$
colors. This bound is tight in the worst case.
Since then many different geometric variants of CF-colorings
have been studied. For example, Har-Peled
and Smorodinsky~\cite{harpeled-cf-05} generalized the
result to objects with near-linear union complexity, while
Even~\etal~\cite{even-cf-03} considered the dual version of the
problem. See the survey by Smorodinsky~\cite{S-survey-10} for an overview.
A restricted type of a CF-coloring is a \emph{unique-maximum} (\emph{UM}) \emph{coloring},
in which the colors are identified with integers, and the
maximum color in the set~$S_p$ is required to be unique. Another type of coloring,
often used as an intermediate step to obtain a CF-coloring, is
\emph{non-monochromatic}~(\emph{NM}). In an NM-coloring---sometimes
called \emph{a proper coloring}---we only require that, for any
point~$p$ in the plane, if the set~$S_p$ contains at least two
elements, not all of them have the same color.
Smorodinsky \cite{smor-geomCF-06} showed that if an
NM-coloring of~$k$ elements using~$\beta(k)$ colors exists for
every~$k$, one can CF-color $n$ elements with~$O(\beta(n) \log n)$ colors.

CF- or NM-coloring objects in $\Reals^1$ 
is significantly easier than in the planar case. In $\Reals^1$ the objects become
intervals, assuming we require the objects to be connected, and a folklore result
states that any set of intervals in $\Reals^1$ can be CF-colored with three
colors and NM-colored with two colors. (This is achieved by the
\emph{chain methods}, which we describe below.) Thus,
unlike in the planar case, the number
of colors for a CF- or NM-coloring of intervals in $\Reals^1$ does not depend
on the number of intervals to be colored.
\medskip

We are interested in generalizations of this result to 1-dimensional
spaces that have a more complex topology than~$\Reals^1$. To this end we consider
\emph{network spaces}: 1-dimensional spaces with the topology of an arbitrary
graph. It is convenient to view a network space~$\graph$ as being embedded
in~$\Reals^2$, although the embedding is actually immaterial.
In this view the \emph{nodes} of~$\graph$ are points in $\Reals^2$,
and the \emph{edges} are simple curves connecting pairs of nodes and otherwise disjoint.
We let~$d\colon\graph^2 \to \reals_+$ denote the geodesic distance on $\graph$.
In other words, for two points $p,q\in \graph$---these points
may lie in the interior of an edge---we let $d(p,q)$ denote the
minimum Euclidean length of any path connecting $p$ to $q$ in $\graph$.
We consider two special types of network spaces, \emph{tree spaces}
and \emph{planar network spaces}, whose topology is that of a
tree and a planar graph, respectively.

The objective of our paper is to investigate the number of colors needed to
CF- or NM-color a set~$\objects$ of $n$ objects in a network space, where we
consider various classes of connected objects. (Here CF- and NM-colorings
are defined as above: in a CF-coloring, for any point $p\in \graph$ the
set $S_p := \{ o\in \Objects \mid p\in o \}$ of objects containing~$p$
should have an object with a unique color when it is non-empty, and in an NM-coloring
the set $S_p$ should not be monochromatic when it consists of at least two objects.)
In particular, we consider balls on $\graph$---the~\emph{ball centered at~$p\in\graph$
of radius~$r$} is defined as $B(p,r):= \{q\in\graph \mid d(p,q)\leq r\}$--- and, for tree spaces,
we also consider arbitrary connected subsets as objects.
Note that, if the
given network space is a single curve, then our setting, both for balls
and for connected subspaces, reduces to
coloring intervals in~$\Reals^1$.  The~main question we want to answer is:
How does the maximum number of colors needed to NM- or CF-color a set $\objects$
of objects in a network space depend on the complexity of the
network space and of the objects to be colored?

\mypara{Our results.}
We assume without loss of generality that the nodes in our network space
either have degree~1 or degree at least~3---there are no nodes of degree~2.
Nodes of degree~1 are also called \emph{leaves}, and nodes of degree at least~3
are also called \emph{internal nodes}.

We start by considering colorings on a tree space, which we denote by~$\treespace$.
Let~$\Objects$ be the set of $n$ objects that we wish
to color, where each object $T\in \Objects$ is a connected subset
of $\treespace$. Note that each such object is itself also a tree.
From now on we refer to the objects in~$\Objects$
as ``trees,'' and always use ``tree space'' when talking about~$\treespace$.
Observe that internal nodes of a tree are necessarily internal nodes
of~$\treespace$, but a tree leaf may lie in
the interior of an edge of~$\treespace$.
We will investigate CF- and NM-chromatic number of trees on tree space
as a function of the following parameters:
\begin{itemize}
\item $k$, the number of leaves of the tree space~$\treespace$;
\item $\ell$, the maximum number of leaves of any tree in $\objects$;
\item $n$, the number of objects in $\objects$.
\end{itemize}
We define the CF-chromatic number~$\CFCN{tree}{trees} (k, \ell; n)$
as the minimum number of colors sufficient to CF-color any set $\Objects$ of $n$ trees of
at most $\ell$ leaves each, in a tree space of at most~$k$ leaves.
The NM-chromatic number~$\NMCN{tree}{trees} (k, \ell; n)$ is defined similarly.
Rows~3 and~4 in Table~\ref{table:overview} give our bounds on these
chromatic numbers. Notice that the upper bounds do not depend on $n$. In other words,
any set of trees in a tree space can be colored with a number of colors
that depends only on the complexity of the tree space~$\treespace$ and 
of the trees in~$\objects$. (Obviously the number of objects, $n$,
is an upper bound on these chromatic numbers as well. To avoid cluttering
the statements, we usually omit this trivial bound.)
\begin{table}
{\small
\begin{center}
\begin{tabular}{l l l r r l}
Space & Objects & \ Coloring & Upper Bound & Lower Bound & \ \ Reference \\[2pt]
\hline \hline \\[-8pt]
Line & Intervals & \ NM & $2$ & $2$ & \ \ Folklore \\[2pt]
\hline \\[-8pt]
Line & Intervals &  \ CF & $3 $ & $3$ & \ \ Folklore \\[2pt]
\hline \hline \\[-8pt]
Tree & Trees & \ NM & $\min\left(\ell +1, 2 \sqrt{6k} \right)$ &
	$\min\left(\ell +1, \left\lfloor \frac{1+ \sqrt{1+8k}}{2}
	\right\rfloor \right)$
	& \ \ Section \ref{sec:gen_subgraph_tree} \\[2pt]
\hline \\[-8pt]
Tree & Trees & \ CF & $O(\ell \log k)$ &
	$\left\lfloor \log_2 \min (k,n) \right\rfloor$
	& \ \ Section \ref{sec:gen_subgraph_tree} \\[2pt]
\hline \hline \\[-8pt]
Tree & Balls & \ NM & $2$ & $2$
	& \ \ Section \ref{sec:balls-on-trees}\\[2pt]
\hline \\[-8pt]
Tree & Balls & \ CF & $\lceil \log t \rceil +3 $
	& $\lceil \log (t+1) \rceil  $
	& \ \ Section \ref{sec:balls-on-trees} \\[2pt]
\hline \hline \\[-8pt]
Planar & Balls & \ NM & $4$ & $4$
	& \ \ Section \ref{sec:balls-on-networks}\\[2pt]
\hline \\[-8pt]
Planar & Balls & \ CF & $\lceil \log_{4/3} t \rceil +3 $ &
	$\lceil \log (t+1) \rceil  $
	& \ \ Section \ref{sec:balls-on-networks} \\[2pt]
\hline
\end{tabular}
\end{center}
\caption{Overview of our results. The folklore result for intervals on the line (that is, in $\Reals^1$) is explained below.} \label{table:overview}
}
\end{table}
We also study balls in tree spaces. Here it turns out to be more convenient
to not use $k$ (the number of leaves) as the complexity measure of $\treespace$,
but
\begin{itemize}
\item $t$, the number of internal nodes of~$\treespace$.
\end{itemize}
We are interested in the chromatic numbers~$\CFCN{tree}{balls} (t; n)$
and $\NMCN{tree}{balls} (t; n)$. Rows~5 and~6 of Table~\ref{table:overview}
state our bounds for these chromatic numbers.

After studying balls in tree spaces, we turn our attention to balls in
planar network spaces. Our bounds on the corresponding chromatic numbers
$\CFCN{planar}{balls} (t; n)$ and $\NMCN{planar}{balls} (t; n)$ 
are contained in row 7 and 8 of Table~\ref{table:overview}.

\mypara{Related results.}
Above we considered CF- and NM-colorings in a geometric setting,
but they can also be defined more abstractly.
A CF-coloring on a hypergraph $\mathcal{H}=(V,E)$ is a coloring of the
vertex set~$V$ such that, for every (non-empty) hyperedge $e\in E$,
there is a vertex in $e$ whose color is different from that of the other
vertices in~$e$. In a NM-coloring any hyperedge with at least two vertices should
not be monochromatic.  Smorodinsky's survey~\cite{S-survey-10}
also gives an overview of results on CF-colorings in this abstract setting.

The basic geometric version mentioned
above---coloring objects in $\Reals^2$ with respect to points---can be phrased in terms
of hypergraphs by letting the objects be the node set $V$
and, for each point $p$ in the plane, creating a hyperedge $e:=S_p$.
Another avenue for constructing a hypergraph~$\mathcal{H}$ to be colored
is to start with a graph~$\graph$, let the vertices of $\mathcal{H}$ be the nodes of $\graph$ and
create hyperedges for (the sets of vertices of) certain subgraphs of~$\graph$.
For example, Pach and Tardos~\cite{PT-cf-09} considered the case
where hyperedges are all the node neighborhoods.
For this case, Abel~\etal~\cite{abel-etal-17} recently showed that a planar
graph can always be CF-colored with only three colors, if we allow
some nodes to be uncolored. (Otherwise, we can use
a dummy color, increasing the number of colors to four.)
As another example, we let the hyperedges be induced by all the paths in the graph. This
setting is equivalent to an older notion of \emph{node ranking} \cite{bod-vertexrank-94},
or \emph{ordered coloring}~\cite{katch-orderedcol-95}.
Note that in the above results the goal is to color the nodes of a graph.
We, on the other hand, do not want to color nodes,
but objects (connected subsets) in a network space (which has a graph topology,
but is a geometric object).

\mypara{Preliminaries: the chain methods.}
We start by describing a folklore technique, called the \emph{chain method}, to
color intervals in~$\Reals^1$ in a non-monochromatic fashion using at most two colors.
We order the intervals left-to-right by their
left endpoints (in case of ties, we take the longest interval first)
and color them in this order using the
so-called \emph{active color} which is defined as follows.
We start with blue as the active color. We color the first interval,
then change the active color to red. We then use the following
procedure: we color the next interval $I$ in the ordering
using the active color,
then if the right endpoint of $I$ is not contained in any
other already colored interval, we change the active color from red to blue
or blue to red.

To obtain a CF-coloring the chain method proceeds as follows.
First, the interval with the leftmost left endpoint---in case of ties,
the longest such interval---is colored blue.
Next, the following procedure is repeated until we get stuck:
Let $I$ be the interval colored last.
Among all intervals whose left endpoint lies in~$I$ and that are not
contained in it, color the one extending farthest to the right red
(if $I$ is blue) or blue (if $I$ is red). This creates a chain
of alternating blue and red intervals.
Each remaining interval is now either completely covered
by the already colored intervals, or it lies completely to the right of them.
The former intervals are given a dummy color (grey),
the latter intervals are colored by applying the above procedure again.

\begin{lemma}\label{lem:chain-methods}
There is a NM-coloring of intervals on a line using two colors,
and a CF-coloring using three colors.
\end{lemma}
\begin{proof}
We prove the latter coloring is conflict-free; the proof for
the NM-coloring is similar. Consider a point~$p$
contained in an interval. It is clear that~$p$ is contained in
either a red or a blue interval. We suppose without loss of generality
it is contained in a red interval~$I_0=[a_0,b_0]$. We show it is not contained in
another red interval. Let us suppose by contradiction that it is
contained in another red interval~$I_1=[a_1,b_1]$ with~$a_1\geqslant a_0$.
Then~$p$ must also be contained in a blue interval~$I_2=[a_2,b_2]$,
with~$a_1 \geqslant a_2 \geqslant a_0$. Moreover, we have
that~$b_2 < b_1 $. Thus,~$I_2$ starts in~$I_0$ and extends further
than~$I_1$, hence should have been chosen to be colored blue, which
is a contradiction. Therefore,~$p$ is always contained in at most
one red interval, and similarly, in at most one blue interval,
and is always contained in a blue or in a red interval.
Thus the coloring is conflict-free.
\end{proof}

\section{Trees on Tree Spaces}\label{sec:gen_subgraph_tree}
\subsection{The upper bound}
\mypara{Overview of the coloring procedure.}
Let~$\TreeSpace$ be a tree space with~$k$ leaves and let~$\Objects$ be a
set of~$n$ trees in~$\treespace$, each with at most~$\ell$ leaves.
We describe an algorithm that NM-colors $\objects$ in two phases:
first, we select
a subset~$\core\subseteq \Objects$ of size at most~$6k-12$
and color it with at most~$\min\left(\ell +1, 2 \sqrt{6k} \right)$ colors.
In the second phase we extend this coloring to the whole set~$\objects$
without using new colors.

An edge $e$ of $\treespace$ is a \emph{leaf edge} if it is incident to a leaf; the remaining edges are \emph{internal}.
We define $\core\subseteq \Objects$ as the set of at most~$6k-12$ trees
selected as follows. For every pair $(e,v)$, where $e$ is an edge of~$\TreeSpace$ and
$v$ is an endpoint of $e$ that is not a leaf of~$\TreeSpace$,
we choose two trees containing~$v$ and extending the furthest into~$e$ (if they exist),
that is, trees~$T$ of $\Objects$ containing~$v$ for which~$\mbox{length}(T\cap e)$ is maximal,
and place them in~$\objects(e,v)$.
If two or more trees of $\Objects$ fully contain~$e$, then~$\objects(e,v)$
contains two of them, chosen arbitrarily.
If a tree contains an internal edge~$e$ fully, it may be chosen by
both endpoints. We now
define~$\objects(e):=\objects(e,u) \cup \objects(e,v)$ for each
internal edge~$e=uv$, $\objects(e):= \objects(e,v)$ for each
leaf edge~$e=uv$ with non-leaf endpoint~$v$,
and $\core:= \bigcup \objects(e)$, with the union taken over all edges~$e$ of~$\treespace$.
Then~$\objects(e)$ contains at most four trees for any internal edge~$e$ and
at most two trees for any leaf edge~$e$.
If $\treespace$ has at most $k$ leaves, it has at most~$k$ leaf edges and at most
$k-3$ internal edges; recall that $\treespace$ has no degree-two nodes. Thus $|\core|\leq 6k-12$, as claimed.
We first explain how to color~$\core$.

\mypara{Coloring~$\core$.}
We color~$\core$ in two steps. Let~$T \in \core$ be a tree.
We define $E(T)$ to be the set of edges~$e$ of~$\treespace$
with~$T\in \objects(e)$.
Firstly, if~$\ell > 2 \sqrt{6k}$
we select all subtrees~$T$ with~$|E(T)| \geqslant \sqrt{6k}$,
and give each of them a unique color.
Since $\sum_{e} |\objects(e)| \leq 6k-12$, there are at most~$\sqrt{6k}-1$
such trees, so we use at most~$\sqrt{6k}-1$ colors.
For each uncolored~$T\in\core$, we create a new
tree~$T'$, defined as the smallest tree
containing~$\bigcup_{e\in E(T)} e\cap T$; see Fig.~\ref{fig:trimming-trees}.
$T'$ has at most~$\ell':=\min(\ell, \sqrt{6k})$ leaves
because $|E(T)|<\sqrt{6k}$.
Define~$\core':=\{T'\mid T\in \core \}$.
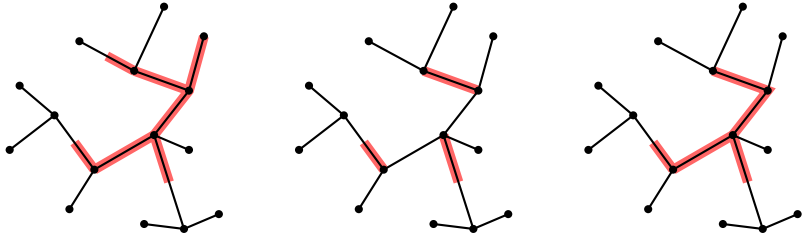
\begin{figure}
\begin{center}
\begin{tikzpicture}[scale=0.655]
\node at (0,0) (1) {};
\node at (1,2) (2) {};
\node at (-0.4,1.3) (3) {};
\node at (-1.2,-0.7) (4) {};
\node at (0.6,-1.9) (5) {};
\node at (-0.2,-1.8) (12) {};
\node at (1.3,-1.6) (13) {};
\node at (0.7,0.9) (6) {};
\node at (0.7,-0.3) (7) {};
\node at (-2,0.4) (8) {};
\node at (-2.9,-0.3) (14) {};
\node at (-2.7,1) (15) {};
\node at (-1.7,-1.5) (9) {};
\node at (-1.5,1.9) (10) {};
\node at (0.2,2.6) (11) {};

\draw[line width=1.3mm, red!60] (1.center) -- (0.3,-0.95);
\draw[line width=1.3mm, red!60] (1.center) -- (4.center);
\draw[line width=1.3mm, red!60] (1.center) -- (6.center) -- (2.center);
\draw[line width=1.3mm, red!60] (3.center) -- (-0.95,1.6);
\draw[line width=1.3mm, red!60] (4.center) -- (-1.6,-0.15);
\draw[line width=1.3mm, red!60] (6.center) -- (3.center);

\draw[thick] (4.center) -- (1.center) -- (5.center);
\draw[thick] (1.center) -- (6.center) -- (2.center);
\draw[thick] (8.center) -- (4.center) -- (9.center);
\draw[thick] (11.center) -- (3.center) -- (10.center);
\draw[thick] (6.center) -- (3.center);
\draw[thick] (7.center) -- (1.center);
\draw[thick] (12.center) -- (5.center) -- (13.center);
\draw[thick] (15.center) -- (8.center) -- (14.center);

\foreach \i in {1,...,15}{
	\draw[thick, fill=black] (\i.center) circle (0.06);
}
\end{tikzpicture}
\qquad
\begin{tikzpicture}[scale=0.655]
\node at (0,0) (1) {};
\node at (1,2) (2) {};
\node at (-0.4,1.3) (3) {};
\node at (-1.2,-0.7) (4) {};
\node at (0.6,-1.9) (5) {};
\node at (-0.2,-1.8) (12) {};
\node at (1.3,-1.6) (13) {};
\node at (0.7,0.9) (6) {};
\node at (0.7,-0.3) (7) {};
\node at (-2,0.4) (8) {};
\node at (-2.9,-0.3) (14) {};
\node at (-2.7,1) (15) {};
\node at (-1.7,-1.5) (9) {};
\node at (-1.5,1.9) (10) {};
\node at (0.2,2.6) (11) {};

\draw[line width=1.3mm, red!60] (1.center) -- (0.3,-0.95);
\draw[line width=1.3mm, red!60] (6.center) -- (3.center);
\draw[line width=1.3mm, red!60] (4.center) -- (-1.6,-0.15);

\draw[thick] (4.center) -- (1.center) -- (5.center);
\draw[thick] (1.center) -- (6.center) -- (2.center);
\draw[thick] (8.center) -- (4.center) -- (9.center);
\draw[thick] (11.center) -- (3.center) -- (10.center);
\draw[thick] (6.center) -- (3.center);
\draw[thick] (7.center) -- (1.center);
\draw[thick] (12.center) -- (5.center) -- (13.center);
\draw[thick] (15.center) -- (8.center) -- (14.center);

\foreach \i in {1,...,15}{
	\draw[thick, fill=black] (\i.center) circle (0.06);
}
\end{tikzpicture}
\qquad
\begin{tikzpicture}[scale=0.655]
\node at (0,0) (1) {};
\node at (1,2) (2) {};
\node at (-0.4,1.3) (3) {};
\node at (-1.2,-0.7) (4) {};
\node at (0.6,-1.9) (5) {};
\node at (-0.2,-1.8) (12) {};
\node at (1.3,-1.6) (13) {};
\node at (0.7,0.9) (6) {};
\node at (0.7,-0.3) (7) {};
\node at (-2,0.4) (8) {};
\node at (-2.9,-0.3) (14) {};
\node at (-2.7,1) (15) {};
\node at (-1.7,-1.5) (9) {};
\node at (-1.5,1.9) (10) {};
\node at (0.2,2.6) (11) {};

\draw[line width=1.3mm, red!60] (1.center) -- (0.3,-0.95);
\draw[line width=1.3mm, red!60] (1.center) -- (4.center);
\draw[line width=1.3mm, red!60] (1.center) -- (6.center) -- (3.center);
\draw[line width=1.3mm, red!60] (4.center) -- (-1.6,-0.15);

\draw[thick] (4.center) -- (1.center) -- (5.center);
\draw[thick] (1.center) -- (6.center) -- (2.center);
\draw[thick] (8.center) -- (4.center) -- (9.center);
\draw[thick] (11.center) -- (3.center) -- (10.center);
\draw[thick] (6.center) -- (3.center);
\draw[thick] (7.center) -- (1.center);
\draw[thick] (12.center) -- (5.center) -- (13.center);
\draw[thick] (15.center) -- (8.center) -- (14.center);

\foreach \i in {1,...,15}{
	\draw[thick, fill=black] (\i.center) circle (0.06);
}
\end{tikzpicture}
\end{center}
\caption{The original tree~$T$ (left), the set~$\bigcup_{e\in E(T)} e\cap T$ (middle),
         and the new tree~$T'$ (right). }\label{fig:trimming-trees}
\end{figure}

The second step is to color~$\core'$.
We need the following lemma, which shows that
an NM-coloring of~$\core'$ carries over to~$\core$.

\begin{lemma}\label{lem:c'_to_c}
Any NM-coloring of~$\core'$ corresponds to an NM-coloring of~$\core$,
that is, if we give each tree~$T\in \core$ the color of the corresponding tree
$T'\in \core'$ then we obtain an NM-coloring.
\end{lemma}
\begin{proof}
Let~$q$ be a point on an edge~$e$ of~$\TreeSpace$ contained in at least two trees
of~$\core$ (if no such trees exists, the coloring is trivially non-monochromatic at~$q$).
Since~$q$ is contained in at least two trees of~$\core$, it is
also contained in two trees of~$\objects(e)$. Call these trees $T_1$ and~$T_2$.
Note that $T_1$ either receives a color in the first coloring step---namely,
when $\ell> 2\sqrt{6k}$ and $|E(T_1)|\geq\sqrt{k}$---or $T'_1\in\core'$ contains~$q$, since $e\in E(T_1)$.
A similar statement holds for~$T_2$. Since the colors used in the first step
are unique and $\core'$ is NM-colored, this implies that $T_1$ and $T_2$ have different colors.
Hence,~$\core$ is NM-colored.
\end{proof}

Next we show how to NM-color~$\core'$. Fix an arbitrary internal
node~$r$ of $\treespace$ and treat $\treespace$ as rooted at $r$.. Our coloring procedure for~$\core'$ maintains
the following invariant: any path from~$r$ to a leaf~$v$ of~$\treespace$ consists
of three disjoint consecutive subpaths (some possibly empty), in this order, as illustrated in Fig.~\ref{fig:3-paths}:
	\begin{itemize}

	\item a \emph{non-monochromatic} subpath containing the root on which at least
	two trees are colored with at least two different colors,

	\item a \emph{singly-colored} subpath covered by exactly one colored tree, and

	\item an \emph{uncolored} subpath containing the leaf on which no tree is
	colored.
	
	\end{itemize}
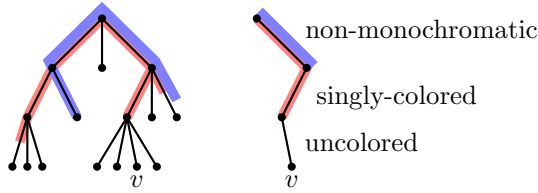
\begin{figure}
\begin{center}
\begin{tikzpicture}[scale=0.655]
\node at (0,0) (0) {};

\node at (1,-1) (1) {};
\node at (-1,-1) (2) {};
\node at (0,-1) (15) {};

\node at (-1.5,-2) (3) {};
\node at (-0.5,-2) (4) {};
\node at (0.5,-2) (5) {};
\node at (1.5,-2) (6) {};
\node at (1,-2) (13) {};

\node at (-1.8,-3) (7) {};
\node at (-1.5,-3) (8) {};
\node at (-1.2,-3) (14) {};
\node at (-0.1,-3) (9) {};
\node at (0.3,-3) (10) {};
\node at (0.7,-3) (11) {};
\node at (1.1,-3) (12) {};

\node at ([yshift=-3mm]11.center) {$v$};

\draw[line width=1.3mm, red!50] (0.center) -- (1.center) -- (5.center);
\draw[line width=1.3mm, red!50] (0.center) -- (2.center) -- (3.center)
								-- (-1.65,-2.5);
\draw[line width=1.3mm, red!50] (1.center) -- (1.25,-1.5);

\draw[line width=1.3mm, blue!50] (4.center)
				-- ([xshift=-0.7mm,yshift=1.4mm]2.center)
				-- ([yshift=2mm]0.center)
				-- ([xshift=0.7mm,yshift=1.4mm]1.center)
				-- ([yshift=3.5mm]6.center);

\foreach \i in {0,...,15}{
	\draw[thick, fill=black ] (\i) circle (0.06);
}

\draw[thick] (0.center) -- (1.center);
\draw[thick] (0.center) -- (2.center);
\draw[thick] (0.center) -- (15.center);

\draw[thick] (1.center) -- (5.center);
\draw[thick] (1.center) -- (6.center);
\draw[thick] (2.center) -- (3.center);
\draw[thick] (2.center) -- (4.center);
\draw[thick] (1.center) -- (13.center);

\draw[thick] (3.center) -- (7.center);
\draw[thick] (3.center) -- (8.center);
\draw[thick] (5.center) -- (9.center);
\draw[thick] (5.center) -- (10.center);
\draw[thick] (5.center) -- (11.center);
\draw[thick] (5.center) -- (12.center);
\draw[thick] (3.center) -- (14.center);
\end{tikzpicture}
\qquad
\begin{tikzpicture}[scale=0.655]
\node at (0,0) (0) {};
\node at (1,-1) (1) {};
\node at (0.5,-2) (4) {};
\node at (0.7,-3) (10) {};

\node at (3.3,-0.2) {non-monochromatic};
\node at (2.8,-1.6) {singly-colored};
\node at (2.1,-2.5) {uncolored};

\node at ([yshift=-3mm]10.center) {$v$};

\draw[line width=1.3mm, red!50] (0.center) -- (1.center) -- (4.center);

\draw[line width=1.3mm, blue!50] ([xshift=0.5mm,yshift=1.5mm]0.center)
				-- ([xshift=1.5mm,yshift=0.5mm]1.center);

\foreach \i in {0,1,4,10}{
	\draw[thick, fill=black ] (\i) circle (0.06);
}

\draw[thick] (0.center) -- (1.center) -- (4.center) -- (10.center);
\end{tikzpicture}
\end{center}
\caption{A coloring of trees (left) and an illustration of the invariant for~$v$ (right). }\label{fig:3-paths}
\end{figure}

\begin{observation}
  Any set of trees containing~$r$ and satisfying the invariant described
  above is NM-colored if we disregard uncolored trees.
\end{observation}

We color the trees $T\in\core'$ that contain~$r$ in an arbitrary order, using~$\ell'+1$
colors, as follows: for each leaf~$v$ of $T$, we follow the path from~$v$ to
the root~$r$ to find a singly-colored part. Note that if we find a singly-colored
part---by the invariant there is at most one such part on the path from $v$ to $r$---we
cannot use that color for~$T$. Since $T$ has at most $\ell'$ leaves, this eliminates at most $\ell'$ colors. Hence, at least one color remains for~$T$.

\begin{lemma}\label{lem:c'}
The procedure described above maintains the invariant and colors
all trees of~$\core'$ containing~$r$ with at most~$\ell'+1$ colors.
\end{lemma}
\begin{proof}
Suppose the invariant holds before the coloring of~$T$. Then we need
to make sure the invariant still holds after~$T$ has been colored.
Let~$w$ be a leaf of~$\treespace$ and~$\pi_w$ the path from~$w$
to the root. Let~$v$ be the closest point to~$w$ in~$\pi_w \cap T$.
Note that~$v$ always exists as~$r \in \pi_w \cap T$.
Now let $\pi_v\subseteq \pi_w$ be the path from~$v$ to~$r$.
It is obvious that~$\pi_w \cap T = \pi_v$.
Then the part of~$\pi_v$ that was uncolored (if it was non-empty) now
is singly-colored. The part that was singly-colored now becomes
non-monochromatic, as we eliminated that color for~$T$. And
the part that was already non-monochromatic stays so.
Therefore the invariant is indeed maintained for~$\pi_w$, concluding
the proof.
\end{proof}

Once all the trees containing~$r$ are colored we delete~$r$ from~$\treespace$, that is, we
consider the space $\treespace\setminus\{r\}$, and we take the closures of the resulting
connected components.
This creates a number of subspaces such that each
uncolored tree in $\core'$ is contained in exactly one of them.
Consider such a subspace $\treespace'$ and let $r'$ be the neighbor of $r$ in
$\treespace'$. We now want to recursively color the uncolored trees in $\treespace'$,
taking~$r'$ as the root of $\treespace'$. However, the invariant might not
hold on the edge~$e$ from~$r'$ to the old root~$r$:
Since now~$r$ is considered a child of~$r'$, the order of the
three parts might switch on~$e$---see
Fig.~\ref{fig:invariant-order-switch}. Suppose this is the
case, and let~$c_e$ be the color of the singly-colored part
on the edge~$e$.
(If the singly-colored part is empty,
we can cut the tree between the non-monochromatic and the uncolored
part and recurse immediately, which maintains the invariant.)
Note also that, for the order to switch, the
non-monochromatic part needs to end on~$e$, and therefore
the only color used in any singly-colored part of the tree
rooted at~$r'$ is~$c_e$.
We overcome this problem by carefully choosing the order in which we color the
trees containing~$r'$. Namely, we fist color the tree~$T$ extending
the farthest into~$e$. In this case, there is only one color
forbidden, namely~$c_e$. We can therefore easily color~$T$.
We can then trim the treespace $\treespace'$ to remove any
non-monochromatic and singly-colored part and
hence restore the invariant and continue with the coloring.
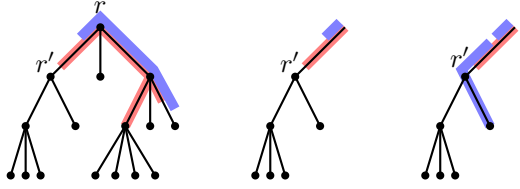
\begin{figure}
\begin{center}
\begin{tikzpicture}[scale=0.655]

\node at (0,0) (0) {};

\node at (1,-1) (1) {};
\node at (-1,-1) (2) {};
\node at (-0.3,-0.3) (2b) {};
\node at (-0.8,-0.8) (2c) {};
\node at (0,-1) (15) {};

\node at (-1.5,-2) (3) {};
\node at (-0.5,-2) (4) {};
\node at (0.5,-2) (5) {};
\node at (1.5,-2) (6) {};
\node at (1,-2) (13) {};

\node at (-1.8,-3) (7) {};
\node at (-1.5,-3) (8) {};
\node at (-1.2,-3) (14) {};
\node at (-0.1,-3) (9) {};
\node at (0.3,-3) (10) {};
\node at (0.7,-3) (11) {};
\node at (1.1,-3) (12) {};

\node at ([yshift=4mm]0.center) {$r$};
\node at ([xshift=-1mm, yshift=3mm]2.center) {$r'$};

\draw[line width=1.3mm, red!50] (0.center) -- (1.center) -- (5.center);
\draw[line width=1.3mm, red!50] (2c.center)
				-- (0.center);
\draw[line width=1.3mm, red!50] (1.center) -- (1.25,-1.5);

\draw[line width=1.3mm, blue!50]
				([xshift=-1mm,yshift=1mm]2b.center)
				-- ([yshift=2mm]0.center)
				-- ([xshift=0.7mm,yshift=1.4mm]1.center)
				-- ([yshift=3.5mm]6.center);

\foreach \i in {0,...,15}{
	\draw[thick, fill=black ] (\i) circle (0.06);
}

\draw[thick] (0.center) -- (1.center);
\draw[thick] (0.center) -- (2.center);
\draw[thick] (0.center) -- (15.center);

\draw[thick] (1.center) -- (5.center);
\draw[thick] (1.center) -- (6.center);
\draw[thick] (2.center) -- (3.center);
\draw[thick] (2.center) -- (4.center);
\draw[thick] (1.center) -- (13.center);

\draw[thick] (3.center) -- (7.center);
\draw[thick] (3.center) -- (8.center);
\draw[thick] (5.center) -- (9.center);
\draw[thick] (5.center) -- (10.center);
\draw[thick] (5.center) -- (11.center);
\draw[thick] (5.center) -- (12.center);
\draw[thick] (3.center) -- (14.center);
\end{tikzpicture}
\qquad
\begin{tikzpicture}[scale=0.655]
\node at (0,0) (0) {};

\node at (-1,-1) (2) {};
\node at (-0.3,-0.3) (2b) {};
\node at (-0.8,-0.8) (2c) {};
\node at (0,-1) (15) {};

\node at (-1.5,-2) (3) {};
\node at (-0.5,-2) (4) {};

\node at (-1.8,-3) (7) {};
\node at (-1.5,-3) (8) {};
\node at (-1.2,-3) (14) {};

\node at ([yshift=4mm]0.center) {};
\node at ([xshift=-1mm, yshift=3mm]2.center) {$r'$};

\draw[line width=1.3mm, red!50] (2c.center)
				-- (0.center);

\draw[line width=1.3mm, blue!50]
				([xshift=-1mm,yshift=1mm]2b.center)
				-- ([xshift=-1mm, yshift=1mm]0.center);

\foreach \i in {2,3,4,7,8,14}{
	\draw[thick, fill=black ] (\i) circle (0.06);
}

\draw[thick] (0.center) -- (2.center);

\draw[thick] (2.center) -- (3.center);
\draw[thick] (2.center) -- (4.center);

\draw[thick] (3.center) -- (7.center);
\draw[thick] (3.center) -- (8.center);
\draw[thick] (3.center) -- (14.center);
\end{tikzpicture}
\qquad
\begin{tikzpicture}[scale=0.655]
\node at (0,0) (0) {};

\node at (-1,-1) (2) {};
\node at (-0.3,-0.3) (2b) {};
\node at (-0.8,-0.8) (2c) {};
\node at (0,-1) (15) {};

\node at (-1.5,-2) (3) {};
\node at (-0.5,-2) (4) {};

\node at (-1.8,-3) (7) {};
\node at (-1.5,-3) (8) {};
\node at (-1.2,-3) (14) {};

\node at ([yshift=4mm]0.center) {};
\node at ([xshift=-1mm, yshift=4mm]2.center) {$r'$};

\draw[line width=1.3mm, red!50] (2c.center)
				-- (0.center);

\draw[line width=1.3mm, blue!50]
				([xshift=-1mm,yshift=1mm]2b.center)
				-- ([xshift=-1mm, yshift=1mm]0.center);
				
\draw[line width=1.3mm, blue!50] (4.center)
				-- ([xshift=-0.7mm,yshift=1.4mm]2.center)
				-- ([xshift=-2mm,yshift=0mm]2b.center);

\foreach \i in {2,3,4,7,8,14}{
	\draw[thick, fill=black ] (\i) circle (0.06);
}

\draw[thick] (0.center) -- (2.center);

\draw[thick] (2.center) -- (3.center);
\draw[thick] (2.center) -- (4.center);

\draw[thick] (3.center) -- (7.center);
\draw[thick] (3.center) -- (8.center);
\draw[thick] (3.center) -- (14.center);
\end{tikzpicture}
\qquad
\begin{tikzpicture}[scale=0.655]
\node at (0,0) (0) {};

\node at (-1,-1) (2) {};
\node at (-0.3,-0.3) (2b) {};
\node at (-0.8,-0.8) (2c) {};
\node at (0,-1) (15) {};

\node at (-1.5,-2) (3) {};
\node at (-0.5,-2) (4) {};

\node at (-1.8,-3) (7) {};
\node at (-1.5,-3) (8) {};
\node at (-1.2,-3) (14) {};

\node at ([yshift=4mm]0.center) {};
\node at ([xshift=-1mm, yshift=4mm]2.center) {$r'$};
				
\draw[line width=1.3mm, blue!50] (4.center)
				-- ([xshift=-0.7mm,yshift=1.4mm]2.center)
				-- ([xshift=-0.9mm,yshift=1.2mm]2c.center);

\foreach \i in {2,3,4,7,8,14}{
	\draw[thick, fill=black ] (\i) circle (0.06);
}

\draw[thick] (2c.center) -- (2.center);

\draw[thick] (2.center) -- (3.center);
\draw[thick] (2.center) -- (4.center);

\draw[thick] (3.center) -- (7.center);
\draw[thick] (3.center) -- (8.center);
\draw[thick] (3.center) -- (14.center);
\end{tikzpicture}
\end{center}
\caption{When recursing on the subspace rooted at~$r'$ (leftmost), the invariant
does not hold anymore (middle left),
as the parts are switched on the edge
between~$r$ and~$r'$. To remedy this, we first color the tree extending
the farthest into that edge (middle right), starting from~$r'$.
We then trim the tree to fix the invariant (rightmost). }
\label{fig:invariant-order-switch}
\end{figure}

\begin{lemma}\label{lem:c}
$\core$ admits an NM-coloring with~$\min (\ell+1, 2\sqrt{6k})$ colors.
\end{lemma}
\begin{proof}
The fact that the procedure above produces an NM-coloring follows from
Lemmas~\ref{lem:c'_to_c} and~\ref{lem:c'}. When $\ell>2\sqrt{6k}$ we use
$\sqrt{6k}-1$ colors to deal with trees $T$ with $|E(T)|\geq\sqrt{6k}$ and
$\ell'+1\leq \min(\ell,2\sqrt{6k})+1\leq \sqrt{6k}+1$ colors for the other trees, giving
$2\sqrt{6k}$ colors in total. When $\ell\leq 2\sqrt{6k}$ we do not treat
the trees with $|E(T)|\geq\sqrt{6k}$ separately, so we just use
$\ell'+1\leq \min(\ell,\sqrt{6k})+1\leq \ell+1$ colors.
  \end{proof}

\mypara{Extending the coloring from~$\core$ to~$\objects$.}
Let~$\col\colon\core \to \naturals$ be an NM-coloring on~$\core$.
We extend the coloring to~$\objects$ as follows. We start by
coloring all trees in~$\objects \setminus \core$ containing
an internal node of $\treespace$ using
an arbitrary color already used. We then treat all edges in an
abritrary order, coloring all trees
contained in the edge as explained now.

Let~$e=rr'$ be an arbitrary edge of~$\treespace$ and~$\objects^*(e)$
be the set of uncolored trees contained in~$e$. We color~$\objects^*(e)$
as follows. We first color the set of uncolored trees contained
in~$e$ naively using the chain method.
For this we use two new colors, which are used for all chains---we
can re-use
the same two colors for the chains, since trivially the chains
in any two edges $e,e'$ do not interact.
However, we can avoid using two extra colors and re-use the colors
from~$\core$ as explained next.

First, if~$\col$ uses fewer than two colors,
then each node of~$\treespace$ is contained in at most one tree.
We then forget the trivial coloring~$\col$ and use the chain method
from scratch on~$\objects$. We start at a arbitrarily
fixed leaf~$u$ of~$\treespace$, and
for any other leaf~$u'$, we consider the path between~$u$ and~$u'$
and use the chain method on the trees restricted to this path.
Since for any node~$v$, at most one tree contains~$v$, no tree
receives two different colors on two different paths.
Moreover, the
coloring is conflict-free, since any point in~$\treespace$ is contained
in a path from~$u$ to a certain leaf~$u'$.

We may now suppose that~$\col$ uses at least two colors.
Let~$T_r\in \objects(e,r)$ and~$T_{r'}\in \objects(e,r')$,
be the trees extending the farthest into~$e$ (arbitrarily chosen
in case of a tie). Note that these trees might not exist.
Also note that $T_r$ and $T_{r'}$ are not in~$\objects^*(e)$.
We define the following colors.
\begin{itemize}
\item Let~$c_r$ be the color of~$T_r$, if~$T_{r}$ exists,
		and an arbitrary color otherwise.
\item Let~$c_{r'}$ be the color of~$T_{r'}$, if~$T_{r'}$ exists,
		and~$\col(T_{r'})\neq \col(T_{r})$ (if~$T_{r}$ does not exist,
		we assume this is always true), and an arbitrary color different
		from~$c_r$ otherwise.
\end{itemize}
We then do the following.
\begin{enumerate}
	\item[(a)] If~$T_{r}$ fully contains~$e$, we color all trees in~$\objects^*(e)$
	using~$c_{r'}$.
	\item[(b)] If~$T_{r'}$ fully contains~$e$, we color all trees in~$\objects^*(e)$
	using~$c_{r}$.
	\item[(c)] Otherwise, we use the chain method for NM-colorings
	using~$c_r$ and~$c_{r'}$ on~$\objects^*(e) \cup \{T_{r}\} \cup \{T_{r'}\}$. We start from~$r$ with color~$c_r$ so that~$T_{r}$ is
	the first tree colored and keep its color.
	We then check if the color of~$T_{r'}$ changed. If so,
	let~$\core_{r'}\subseteq \core$ be the subset
	of trees contained in the subspace rooted at~$r'$ (including~$e$
	but not~$r$) and excluding~$T_{r'}$. We exchange~$c_r$ and~$c_{r'}$
	in~$\core_{r'}$; see Fig. \ref{fig:case-no-intersection}.
\end{enumerate}

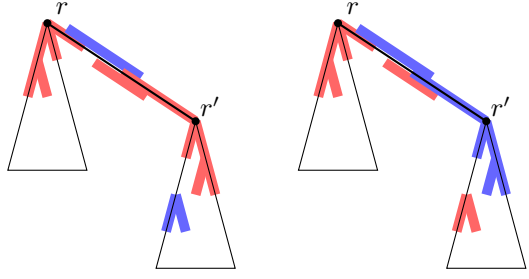
\begin{figure}
\begin{center}
\begin{tikzpicture}[scale=0.65]
\node at (0,0) (r) {};
\node at (3,-2) (r') {};

\draw[line width=1.3mm, red!60] (r.center) --++ (-0.4,-1.5);
\draw[line width=1.3mm, red!60] (r.center) --++ (0.2,-0.75);
\draw[line width=1.3mm, red!60] ([xshift=-0.2cm, yshift=-0.75cm]r.center)  --++ (0.2,-0.75);
\draw[line width=1.3mm, red!60] (r.center) --++ (0.75, -0.5);

\draw[line width=1.3mm, blue!60] ([xshift=0.4cm, yshift=-0.1cm]r.center)  --++ (1.5,-1);

\draw[line width=1.3mm, red!60] ([xshift=0.95cm, yshift=-0.8cm]r.center)  --++ (1.05,-0.7);

\draw[line width=1.3mm, red!60] ([xshift=1.5cm, yshift=-1cm]r.center)  --++ (1.5,-1);
\draw[line width=1.3mm, red!60] (r'.center) --++ (0.4,-1.5);
\draw[line width=1.3mm, red!60] (r'.center) --++ (-0.2,-0.75);
\draw[line width=1.3mm, red!60] ([xshift=0.2cm, yshift=-0.75cm]r'.center)  --++ (-0.2,-0.75);

\draw[line width=1.3mm, blue!60] ([xshift=-0.4cm, yshift=-1.5cm]r'.center)  --++ (0.2,-0.75);
\draw[line width=1.3mm, blue!60] ([xshift=-0.4cm, yshift=-1.5cm]r'.center)  --++ (-0.2,-0.75);

\foreach \i in {r, r'}{
	\draw (\i.center) --++ (-0.8, -3) --++ (1.6,0) -- cycle;
	\draw[thick, fill=black] (\i.center) circle (0.06);
	\node at ([xshift=3mm, yshift=3mm]\i.center) {$\i$};
}

\draw[thick] (r.center) -- (r'.center);
\end{tikzpicture}
\qquad
\begin{tikzpicture}[scale=0.65]
\node at (0,0) (r) {};
\node at (3,-2) (r') {};

\draw[line width=1.3mm, red!60] (r.center) --++ (-0.4,-1.5);
\draw[line width=1.3mm, red!60] (r.center) --++ (0.2,-0.75);
\draw[line width=1.3mm, red!60] ([xshift=-0.2cm, yshift=-0.75cm]r.center)  --++ (0.2,-0.75);
\draw[line width=1.3mm, red!60] (r.center) --++ (0.75, -0.5);

\draw[line width=1.3mm, blue!60] ([xshift=0.4cm, yshift=-0.1cm]r.center)  --++ (1.5,-1);

\draw[line width=1.3mm, red!60] ([xshift=0.95cm, yshift=-0.8cm]r.center)  --++ (1.05,-0.7);

\draw[line width=1.3mm, blue!60] ([xshift=1.5cm, yshift=-1cm]r.center)  --++ (1.5,-1);
\draw[line width=1.3mm, blue!60] (r'.center) --++ (0.4,-1.5);
\draw[line width=1.3mm, blue!60] (r'.center) --++ (-0.2,-0.75);
\draw[line width=1.3mm, blue!60] ([xshift=0.2cm, yshift=-0.75cm]r'.center)  --++ (-0.2,-0.75);

\draw[line width=1.3mm, red!60] ([xshift=-0.4cm, yshift=-1.5cm]r'.center)  --++ (0.2,-0.75);
\draw[line width=1.3mm, red!60] ([xshift=-0.4cm, yshift=-1.5cm]r'.center)  --++ (-0.2,-0.75);

\foreach \i in {r, r'}{
	\draw (\i.center) --++ (-0.8, -3) --++ (1.6,0) -- cycle;
	\draw[thick, fill=black] (\i.center) circle (0.06);
	\node at ([xshift=3mm, yshift=3mm]\i.center) {$\i$};
}

\draw[thick] (r.center) -- (r'.center);
\end{tikzpicture}
\end{center}
\caption{If the color of~$T_{r'}$ changes with the chain method,
we swap the labels of the old and new colors of~$T_{r'}$
in the subspace rooted at~${r'}$.
}\label{fig:case-no-intersection}
\end{figure}
The following lemma proves the extended coloring is non-monochromatic.
\begin{lemma}\label{lem:extending-coloring}
  Any NM-coloring~$c$ on~$\core$ can be extended
  to~$\Objects$ without using any extra color if~$c$ uses
  two colors or more, and with two colors otherwise.
\end{lemma}
\begin{proof}
Let $\objects_1$ be the subset of trees in $\objects\setminus \core$ that contain an internal node of~$\treespace$, and let $\objects_2$ be the remaining trees in $\objects\setminus \core$. By Lemma~\ref{lem:c}, we have an NM-coloring on $\core$.
To prove that the method described above gives us an NM-coloring
on~$\core\cup\objects_2$, we show that the following invariant
holds each time an edge is colored: the coloring
on~$\core\cup\objects_2$ is non-monochromatic when
restricted to colored trees. It is clear that before the first edge
is colored, the coloring is non-monochromatic as at this point the
only trees colored are exactly those in~$\core$. We hence only have
to show the invariant still holds after coloring an edge~$e=\{r,r'\}$.
If we are in cases (a) or (b), the invariant
trivially holds. It remains to consider the third case.

In the case~(c) we use the chain
method on~$\objects^*(e) \cup \{T_{r}\} \cup \{T_{r'}\}$,
which immediately implies the coloring is non-monochromatic on~$e$.
To prove it is also non-monochromatic elsewhere,
let~$p\notin e$ be a point contained in at least
two trees. Then we only have to show that the label swap we did on
one side of~$e$ keeps the coloring non-monochromatic.
The point~$p$ cannot be contained in one tree containing~$r$
and one tree containing~$r'$ at the same time, because no
tree contains~$e$ fully. Therefore,~$p$ is
contained in at least two trees from either side of~$e$, hence
two trees of different color.

Furthermore, the trees in $\objects_1$ received an arbitrary color already used. To prove that this gives an NM-coloring for $\objects = \core \cup \objects_1\cup\objects_2$, it suffices to prove that each tree $T\in\objects_1$ is \emph{doubly-covered} by~$\core$, that is, any point $q\in T$ is contained in at least two trees in~$\core$. To this end, let $e$ be an edge such that $q\in e$. Then, since $T\not\in \core$ and $T$ contains an endpoint $v$ of~$e$, the two trees in $\objects(e,v)$ contain~$q$. Hence, $T$ is doubly-covered by~$\core$, as claimed.
  \end{proof}

\begin{theorem}\mbox{}
  \label{thm:main}
\begin{enumerate}
   	  \item $\NMCN{tree}{trees} (k, \ell; n)
    	    \leq \min\left(\ell +1, 2 \sqrt{6k} \right)$.
      \item $\CFCN{tree}{trees} (k, \ell; n)=O(\ell \log k)$.
  	\end{enumerate}
\end{theorem}
\begin{proof}
  For the NM-coloring part of the theorem,
  we use Lemmas~\ref{lem:c}
  and~\ref{lem:extending-coloring}.
For the second part, if $\ell > 2\sqrt{6k}$ we again reduce $\core$
to $\core'$ using at most $\sqrt{6k}-1$ colors. Then
use the result by Smorodinsky \cite{smor-geomCF-06} on the
NM-coloring on $\core'$ provided by Lemma~\ref{lem:c'}. Since this
coloring uses at most $\ell' +1$ colors and $|\core'|\leqslant 6k-12$, the
CF-coloring uses $O(\ell \log k)$ colors. We then extend the coloring
to $\objects$ using similar techniques as for the NM-coloring.
This coloring uses $O(\sqrt{k} \log k)$ colors if $\ell > 2\sqrt{6k}$,
which is in $O(\ell \log k)$, and directly $O(\ell \log k)$ colors
otherwise. Note that a direct application of the result of
Smorodinsty~\cite{smor-geomCF-06} would give a~$O(\ell \log n)$ bound
instead.
  \end{proof}

\subsection{The lower bound}
We show a lower bound for the number of colors\footnote{From now on, we either identify colors
with integers or we use actual colors (red, blue, etc.) in our descriptions, whichever is more convenient.}
needed to NM-color a set of trees in a tree space.

\begin{theorem}\label{thm:lb-nm}
For all~$n, k$, and~$\ell$, there exist a tree space~$\TreeSpace$
with~$k$ leaves and a set~$\objects$
at most~$n$ trees on~$\treespace$, each with at most~$\ell$
leaves, such that any non-monochromatic coloring of~$\objects$ uses
at least~$\min \left(\ell+1, \left\lfloor \tfrac{1+ \sqrt{1+8k} }{2} \right\rfloor, n \right)$ colors. In other words,
$$\NMCNtreetree (k, \ell; n)\geqslant
\min \left(\ell+1, \left\lfloor \tfrac{1+ \sqrt{1+8k} }{2} \right\rfloor, n \right).
 $$
\end{theorem}
\begin{proof}
Let~$\TreeSpace$ be
a star with~$k$ leaves. We construct the set~$\Objects$ of~$m$
trees such that, for each pair of trees~$T,T'\in \Objects$,
there is a leaf of~$\TreeSpace$ contained in~$T$ and~$T'$, and no
other tree from~$\objects$.
Consequently, each tree in~$\Objects$ must be assigned
a distinct color.
To this end, we define~$m:=\min(\ell+1,m',n)$,
where $m':=\lfloor(1+ \sqrt{1+8k})/2 \rfloor$ is the
largest integer such that~${m'\choose 2} \leqslant k $.
Then, for every pair~$\{i,j\}$ with~$1\leqslant i<j \leqslant m$,
we choose a distinct leaf of~$\TreeSpace$ and associate it
with~$\{i,j\}$. The total number of such pairs
is~${m \choose 2} \leqslant {m' \choose 2} \leqslant k $,
hence we can indeed associate a distinct leaf to each pair.

Let now~$\Objects:= \{ T_1,\ldots, T_m\}$ be the set of
trees defined as follows: for each~$i=1,\ldots,m$, the
tree~$T_i$ is defined as the tree containing
all the leaves associated with pairs~$\{i,j\}$ for some~$j\neq i$,
i.e.,~$T_i$ is the union, for all~$j\neq i$, of
edges from the root to a leaf associated with~$\{i,j\}$.
Fig.~\ref{fig:star_lb} shows an example.
\begin{figure}
\begin{center}
\begin{tikzpicture}
\node at (0,0) (0) {};
\node at (0:1) (1) {};
\node at (60:1) (2) {};
\node at (120:1) (3) {};
\node at (180:1) (4) {};
\node at (240:1) (5) {};
\node at (300:1) (6) {};

\foreach \i in {1,...,3}{
	\draw[line width=1.3mm, red!50] (0.center) -- (\i.center);

\foreach \i in {0,...,6}{
	\draw[thick, fill=black] (\i) circle (0.06);
}
}
\foreach \i in {1,...,6}{
	\draw[thick] (0.center) -- (\i.center);
}

\node at (0:1.7) {$\{1,2\}$};
\node at (60:1.7) {$\{1,3\}$};
\node at (120:1.7) {$\{1,4\}$};
\node at (180:1.7) {$\{2,3\}$};
\node at (240:1.7) {$\{2,4\}$};
\node at (300:1.7) {$\{3,4\}$};
\end{tikzpicture}
\caption{An example of the non-monochromatic lower bound for~$k=6$,
$\ell=3$, and~$n=4$. The tree~$T_1$ is drawn in red. }
\label{fig:star_lb}
\end{center}
\end{figure}
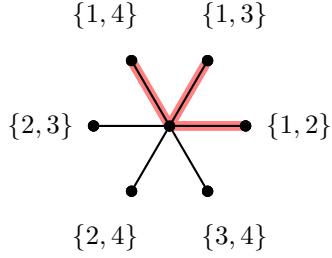
We now have to prove that
the construction is possible within the parameters.
Recall that~$m\leqslant n$ so we have indeed at most~$n$
trees in~$\objects$, and that~$m\leqslant m'$ where~$m'$
is chosen to ensure~$k$ leaves are enough. We therefore only have
to show that no tree~$T_i,\ldots, T_m$
has more than~$\ell$ leaves.
However, the number of leaves of each tree~$T_i$ is at most~$m-1$,
as we only create at most one leaf for~$T_i$ for each~$T_j$
with~$j\neq i$. Hence, since~$m\leqslant \ell +1$, each tree
has at most~$\ell$ leaves. Thus, the construction
does not violate the parameters.

Finally, each tree needs a distinct color, and since
there are~$m$ trees, the number of colors needed
is~$m=\min (\ell+1, \lfloor \tfrac{1+ \sqrt{1+8k} }{2} \rfloor, n)$.
  \end{proof}

Since any CF-coloring is also an NM-coloring, the lower bound in
Theorem~\ref{thm:lb-nm} holds for CF-coloring as well. The next
theorem gives a stronger lower bound for CF-coloring in the case $\ell=2$,
that is, when the objects are paths.

\begin{theorem}
For all~$n$ and~$k$, there exist a tree space~$\TreeSpace$
with~$k$ leaves and a~set~$\objects$
of at most~$n$ paths in~$\treespace$ such that any conflict-free
coloring of~$\objects$ uses
at least~$\lfloor\log_2 \min(k,n)\rfloor$ colors. In other words,
$$\CFCNtreepath (k; n)\geqslant \lfloor\log_2 \min(k,n)\rfloor. $$
\end{theorem}
\begin{proof}
Let~$\TreeSpace$ be a rooted complete binary tree of
height~$h=\lfloor\log_2 \min(k,n)\rfloor$. Note that $\treespace$
has at most $\min(k,n)$ leaves. For
each leaf~$v$ of~$\TreeSpace$, we define~$\pi_v$ to be the path from~$v$
to the root~$r$ of~$\TreeSpace$. Our set~$\Objects$ of objects
is now defined
as~$\Objects:=\{ \pi_v \mid v \text{ a leaf of } \TreeSpace \}$.
(Trivially,~$|\objects|\leqslant n$.)
%
%
%
%
%
%

Let~$c\colon\Objects \to \integers $ be a conflict-free
coloring of~$\Objects$.
We prove that~$c$ uses at least~$h=\lfloor\log_2 \min(k,n)\rfloor$
colors by induction on the height~$h$ of~$\TreeSpace$.
If~$h=1$, then there is only one degenerate path
and the claim trivially holds. Suppose now that the claim holds
for a tree of height~$h$, and suppose the height of~$\TreeSpace$
is~$h+1$. Since~$c$ is a conflict-free
coloring, among the paths containing the root~$r_1:=r$ of~$\TreeSpace$,
there must be a path~$\pi_1$ of unique color. Since by construction
all paths in~$\Objects$ contain the root, the color of~$\pi_1$
is unique among all paths. Let~$r_2$ be the child of~$r_1$
not contained in~$\pi_1$.
We now use the induction hypothesis on the subtree rooted
at~$r_2$ with paths containing~$r_2$ cut above it. Among these
paths, there are~$h$ that use distinct colors. Moreover,
none of these path can use~$c(\pi_1)$, as this color is unique
among all paths. Hence, we have indeed~$h+1$ paths using distinct
colors. This concludes the proof.
  \end{proof}
The following theorem is a direct consequence of the previous two.
\begin{theorem}
For all~$n, k$, and~$\ell$, there exist a tree space~$\TreeSpace$ with~$k$ leaves and a set~$\objects$
at most~$n$ trees in~$\treespace$ with at most~$\ell$ leaves each
such that any conflict-free coloring of~$\objects$ uses
at least~$\min \left(\ell+1, \left\lfloor \tfrac{1+ \sqrt{1+8k} }{2} \right\rfloor,
\lfloor\log_2 \min(k,n)\rfloor\right)$ colors.
In other words,
$$\CFCNtreetree (k, \ell; n)\geqslant \max
\left\{\begin{array}{l}
        \min \left(\ell+1, \left\lfloor \tfrac{1+ \sqrt{1+8k} }{2} \right\rfloor \right) \\[16pt]
        \lfloor\log_2 \min(k,n)\rfloor.
        \end{array} \right.
 $$
\end{theorem}

\section{Balls in Tree Spaces and on Planar Network Spaces}\label{sec:balls}
In this section we restrict the objects to balls.
Let~$\graph$ be a network space,~$d\colon\graph^2 \to \reals$
a distance function on~$\graph$, and let~$\objects$ be a set of balls
on~$\graph$.
We  define the coverage~$\cov_x(B)$ of a node~$x$
by a ball~$B=B(p,r)$ containing~$x$ as~$\cov_x(B):=r-d(p,x)$.
Given a node~$x$ contained in at least one ball from~$\objects$,
we define~$B_x$ as the ball maximizing the coverage of~$x$,
where we break ties using an arbitrary but fixed ordering on the balls.
We say that $B_x$ is \emph{assigned} to~$x$.
Note that~$B_x$ does not exist if no ball contains~$x$, and that
a ball can be assigned
to multiple nodes. We will regularly use the following lemma
regarding the assigned balls.

\begin{lemma}\label{lem:connected_core}
Let $x$ be an internal node of $\graph$.
\begin{enumerate}
\item[(i)] Suppose $\graph$ is a tree space, and let $\tree_1,\ldots,\tree_{{\mathrm deg}(x)}$ denote the
           subtrees resulting from removing $x$ from~$\graph$ or, more precisely,
           the closures of the connected components of $\tree\setminus\{x\}$.
 Let $p$ be a point in some subtree $\tree_i$ and suppose $p$ is contained in a ball $B\in \objects$ whose center lies in
           $\tree_j$ with $j\neq i$. Then $p\in B_x$.
\item[(ii)] Suppose $x$ is contained in at least one ball in $\objects$.
   Let $\pi$ be a shortest path from~$x$ to the center of~$B_x$, and
            let $y$ be a node on the path~$\pi$. Then~$B_x$ is also
            assigned to~$y$, that is,~$B_x=B_y$.
\end{enumerate}
\end{lemma}
\begin{proof}
Part~(i) follows immediately from the definition of~$B_x$. To prove part~(ii),
suppose for a contradiction that~$B_y \neq B_x$ for some $y\in\pi$.
Thus,~$\cov_y(B_y) \geqslant \cov_y(B_x)$. Because~$\pi$ is a
shortest path from $x$ to the center of $B_x$, we
have that~$\cov_x(B_x)  = \cov_y(B_x) - d(x,y)$.
Moreover,~$ \cov_y(B_y) - d(x,y) \leqslant \cov_x(B_y)$ because
of the triangle inequality.
Hence,~$\cov_x(B_x)
	\geqslant \cov_x(B_y)
	\geqslant \cov_y (B_y) -d(x,y)
	\geqslant cov_y(B_x) -d(x,y)
	=\cov_x(B_x) $.
Thus~$ \cov_x(B_x) = \cov_x(B_y)$ and~$\cov_y(B_x) = \cov_y(B_y)$.
However, this is a contradiction as in case of a tie, we use
the fixed ordering to choose which ball to assign to a node.
  \end{proof}

\subsection{Tree spaces: the upper bound}\label{sec:balls-on-trees}
For balls on a tree space~$\treespace$, the upper bounds
from Theorem~\ref{thm:main} with $\ell=k$ apply.
Below we improve upon these bounds using the special structures of balls.
Let~$\treespace$ be a tree with~$t$ internal nodes.
We present algorithms to
NM-color balls on trees using two colors, and
CF-color them with~$\log t +3$ colors.

Let~$\objects$ be a set of~$n$ balls on~$\treespace$.
Let also~$\core:=\{B=B(c,r) \mid \exists x: B=B_x \}$ be
the set of balls assigned to at least one internal node. Recall that
an internal node~$x$ is assigned the ball maximizing the coverage
of~$x$.

\mypara{NM-coloring.}
We first explain how to NM-color~$\objects$. We
use a divide-and-conquer approach. If~$t=0$, that is~$\treespace$
consists of a single node or a single edge, we use the
chain method for~NM-coloring with colors blue and red.
If $t>0$, then we proceed as follows.
Let~$e=uv$ be an edge of~$\treespace$.
Let~$\treespace_u$, respectively~$\treespace_v$, be the
connected component of~$\treespace \setminus e$
containing~$u$, respectively~$v$.
Recall that~$B_u$ and~$B_v$ are the balls assigned to~$u$ and~$v$, respectively.
Note that we may assume that both $B_u$ or~$B_v$ exist, for otherwise
recursion is trivial.
Also observe that~$B_u$ and~$B_v$ may coincide.
We define
\[
\objects(u) := \{ \mbox{balls $B\in \objects$ whose center lies in~$\treespace_u$} \} \cup \{B_u\},
\]
We define~$\objects(v)$ similarly.
We recursively color~$\objects(u)$ in~$\treespace_u$ and~$\objects(v)$ in~$\treespace_v$,
obtaining colorings of~$\objects(u)$ and $\objects(v)$ with colors blue and red.
In the recursive calls on $\objects(u)$, and similarly for $\objects(v)$, we ``clip'' the balls to within $\treespace_u$. Note that the clipped balls are still balls in the space $\treespace_u$. This is clear for the balls whose center lies in $\treespace_u$. The center of $B_u$ may not lie in $\treespace_u$, but in that case it behaves within $\treespace_u$ as a ball  with center $u$ and radius $\cov_u(B_u)$.

Let~$\objects(e):= \objects \setminus (\objects(u)\cup \objects(v))$ be the set of the remaining
balls. In other words, $\objects(e)$ contains the balls whose center
is contained in~$e$, except for~$B_u$ and~$B_v$.
We color~$\objects(e)$, possibly swapping colors
in~$\objects(u)$ or~$\objects(v)$, as follows.
\begin{itemize}
\item If~$B_u=B_v$, we first ensure that it gets the same
	  color in both~$\objects(u)$ and~$\objects(v)$ by swapping
	  colors in one of the two subsets if necessary. We then color
	  all balls in~$\objects(e)$ blue if~$B_u$
	  is red, and red if~$B_u$ is blue.
\item If~$B_u \neq B_v$, let~$\pi$ be a longest simple path
	  containing~$u$ and~$v$. We
	  color~$\objects(e) \cup \{ B_u, B_v \}$ restricted to~$\pi$
	  using the non-monochromatic chain method. We then possibly
	  swap colors in~$\objects(u)$ and~$\objects(v)$ so that~$B_u$
	  and~$B_v$ match the colors they were given by the chain method.
\end{itemize}
Both cases are illustrated in Fig.~\ref{fig:tree-recursive-NM}.
\begin{figure}
\begin{center}
\begin{tikzpicture}[scale=0.65]
\node at (0,0) (1) {};
\node at (2,0.5) (2) {};
\node at (-0.8,1) (3) {};
\node at (-1.2,-0.7) (4) {};
\node at (-0.6,-1.9) (5) {};
\node at (3.2,0.9) (6) {};
\node at (2.9,-0.6) (7) {};
\node at (1.8,1.5) (8) {};

\node at ([xshift=4mm, yshift=-4mm]1.center) {$B_u$};
\node at ([xshift=0mm, yshift=-6mm]2.center) {$B_v$};

\draw[line width=1.3mm, black!30] (-0.3,-0.95) -- (1.center)
					-- (0.8,0.2);
\draw[line width=1.3mm, black!30] (4.center) -- (1.center)
					-- (-0.4, 0.5);
					
\draw[line width=1.3mm, black!30] (0.5, 0.305) -- (1.9, 0.655);
					
\draw[line width=1.3mm, black!30] (2.6, 0.7) -- (2.center)
					-- (1.6,0.4);
\draw[line width=1.3mm, black!30] (2.45,-0.05) -- (2.center);

\draw[line width=1.3mm, black!30] 
				([xshift=-1.2mm, yshift=-1.2mm]3.center) 
					-- (-0.32,0.13);
\draw[line width=1.3mm, black!30] 
				([xshift=-1.5mm, yshift=0.5mm]5.center) 
					-- (-0.35,-0.575);
					
\draw[line width=1.3mm, black!30] 
				(8.center) -- (1.9,1);					
\draw[line width=1.3mm, black!30] 
				([xshift=0.2mm, yshift=-1.7mm]6.center) 
				-- ([xshift=3.2mm, yshift=-1mm]2.center) 
				-- ([xshift=1.5mm, yshift=1mm]7.center);

\draw[thick]  (1.center) -- (2.center);
\draw[thick]  (1.center) -- (3.center);
\draw[thick]  (1.center) -- (4.center);
\draw[thick]  (1.center) -- (5.center);
\draw[thick]  (6.center) -- (2.center);
\draw[thick]  (7.center) -- (2.center);
\draw[thick]  (8.center) -- (2.center);

\foreach \i in {1,...,8}{
	\draw[thick, fill=black] (\i.center) circle (0.06);
}
\end{tikzpicture}
\hspace{0.5cm}
\begin{tikzpicture}[scale=0.65]
\node at (0,0) (1) {};
\node at (2,0.5) (2) {};
\node at (-0.8,1) (3) {};
\node at (-1.2,-0.7) (4) {};
\node at (-0.6,-1.9) (5) {};
\node at (3.2,0.9) (6) {};
\node at (2.9,-0.6) (7) {};
\node at (1.8,1.5) (8) {};

\node at ([xshift=-10mm, yshift=0mm]1.center) {$\objects(u)$};
\node at ([xshift=4mm, yshift=-4mm]1.center) {$B'_u$};
\node at ([xshift=6mm, yshift=7mm]2.center) {$\objects(v)$};
\node at ([xshift=-1mm, yshift=-6mm]2.center) {$B'_v$};

\draw[line width=1.3mm, red!60]  (4.center)-- (1.center);
\draw[line width=1.3mm, blue!60] (2.center) -- (2.6, 0.7);
\draw[line width=1.3mm, red!60] (-0.3,-0.95) -- (1.center)
					-- (-0.4, 0.5);
\draw[line width=1.3mm, blue!60] (2.45,-0.05) -- (2.center);

\draw[line width=1.3mm, blue!60] 
				([xshift=-1.2mm, yshift=-1.2mm]3.center) 
					-- (-0.32,0.13);
\draw[line width=1.3mm, blue!60] 
				([xshift=-1.5mm, yshift=0.5mm]5.center) 
					-- (-0.35,-0.575);
					
\draw[line width=1.3mm, blue!60] 
				(8.center) -- (1.9,1);					
\draw[line width=1.3mm, red!60] 
				([xshift=0.2mm, yshift=-1.7mm]6.center) 
				-- ([xshift=3.2mm, yshift=-1mm]2.center) 
				-- ([xshift=1.5mm, yshift=1mm]7.center);

\draw[thick]  (1.center) -- (3.center);
\draw[thick]  (1.center) -- (4.center);
\draw[thick]  (1.center) -- (5.center);
\draw[thick]  (6.center) -- (2.center);
\draw[thick]  (7.center) -- (2.center);
\draw[thick]  (8.center) -- (2.center);

\foreach \i in {1,...,8}{
	\draw[thick, fill=black] (\i.center) circle (0.06);
}
\end{tikzpicture}
\hspace{0.5cm}
\begin{tikzpicture}[scale=0.65]
\node at (0,0) (1) {};
\node at (2,0.5) (2) {};
\node at (-0.8,1) (3) {};
\node at (-1.2,-0.7) (4) {};
\node at (-0.6,-1.9) (5) {};
\node at (3.2,0.9) (6) {};
\node at (2.9,-0.6) (7) {};
\node at (1.8,1.5) (8) {};

\node at ([xshift=4mm, yshift=-4mm]1.center) {$B_u$};
\node at ([xshift=0mm, yshift=-6mm]2.center) {$B_v$};

\draw[line width=1.3mm, red!60] (4.center) -- (1.center)
					-- (0.8,0.2);
\draw[line width=1.3mm, red!60]  (-0.3,-0.95)-- (1.center)
					-- (-0.4, 0.5);
					
\draw[line width=1.3mm, blue!60] (0.5, 0.305) -- (1.9, 0.655);
					
\draw[line width=1.3mm, red!60] (2.6, 0.7) -- (2.center)
					-- (1.6,0.4);
\draw[line width=1.3mm, red!60] (2.45,-0.05) -- (2.center);

\draw[line width=1.3mm, blue!60] 
				([xshift=-1.2mm, yshift=-1.2mm]3.center) 
					-- (-0.32,0.13);
\draw[line width=1.3mm, blue!60] 
				([xshift=-1.5mm, yshift=0.5mm]5.center) 
					-- (-0.35,-0.575);
					
\draw[line width=1.3mm, red!60] 
				(8.center) -- (1.9,1);					
\draw[line width=1.3mm, blue!60] 
				([xshift=0.2mm, yshift=-1.7mm]6.center) 
				-- ([xshift=3.2mm, yshift=-1mm]2.center) 
				-- ([xshift=1.5mm, yshift=1mm]7.center);

\draw[thick]  (1.center) -- (2.center);
\draw[thick]  (1.center) -- (3.center);
\draw[thick]  (1.center) -- (4.center);
\draw[thick]  (1.center) -- (5.center);
\draw[thick]  (6.center) -- (2.center);
\draw[thick]  (7.center) -- (2.center);
\draw[thick]  (8.center) -- (2.center);

\foreach \i in {1,...,8}{
	\draw[thick, fill=black] (\i.center) circle (0.06);
}
\end{tikzpicture}
\\
\begin{tikzpicture}[scale=0.65]
\node at (0,0) (1) {};
\node at (2,0.5) (2) {};
\node at (-0.8,1) (3) {};
\node at (-1.2,-0.7) (4) {};
\node at (-0.6,-1.9) (5) {};
\node at (3.2,0.9) (6) {};
\node at (2.9,-0.6) (7) {};
\node at (1.8,1.5) (8) {};

\node at ([xshift=12mm, yshift=-3mm]1.center) {$B_u=B_v$};

\draw[line width=1.3mm, black!30] (-0.3,-0.95) -- (1.center)
					-- (2.center) -- (2.6, 0.7);
\draw[line width=1.3mm, black!30] (4.center) -- (1.center)
					-- (-0.4, 0.5);
\draw[line width=1.3mm, black!30] (2.45,-0.05) -- (2.center);

\draw[line width=1.3mm, black!30] (0.5, 0.305) -- (1.9, 0.655);

\draw[line width=1.3mm, black!30] 
				([xshift=-1.2mm, yshift=-1.2mm]3.center) 
					-- (-0.32,0.13);
\draw[line width=1.3mm, black!30] 
				([xshift=-1.5mm, yshift=0.5mm]5.center) 
					-- (-0.35,-0.575);
					
\draw[line width=1.3mm, black!30] 
				(8.center) -- (1.9,1);					
\draw[line width=1.3mm, black!30] 
				([xshift=0.2mm, yshift=-1.7mm]6.center) 
				-- ([xshift=3.2mm, yshift=-1mm]2.center) 
				-- ([xshift=1.5mm, yshift=1mm]7.center);

\draw[thick]  (1.center) -- (2.center);
\draw[thick]  (1.center) -- (3.center);
\draw[thick]  (1.center) -- (4.center);
\draw[thick]  (1.center) -- (5.center);
\draw[thick]  (6.center) -- (2.center);
\draw[thick]  (7.center) -- (2.center);
\draw[thick]  (8.center) -- (2.center);

\foreach \i in {1,...,8}{
	\draw[thick, fill=black] (\i.center) circle (0.06);
}
\end{tikzpicture}
\hspace{0.5cm}
\begin{tikzpicture}[scale=0.65]
\node at (0,0) (1) {};
\node at (2,0.5) (2) {};
\node at (-0.8,1) (3) {};
\node at (-1.2,-0.7) (4) {};
\node at (-0.6,-1.9) (5) {};
\node at (3.2,0.9) (6) {};
\node at (2.9,-0.6) (7) {};
\node at (1.8,1.5) (8) {};

\node at ([xshift=-10mm, yshift=0mm]1.center) {$\objects(u)$};
\node at ([xshift=4mm, yshift=-4mm]1.center) {$B'_u$};
\node at ([xshift=6mm, yshift=7mm]2.center) {$\objects(v)$};
\node at ([xshift=-1mm, yshift=-6mm]2.center) {$B'_v$};

\draw[line width=1.3mm, red!60] (4.center) -- (1.center);
\draw[line width=1.3mm, blue!60] (2.center) -- (2.6, 0.7);
\draw[line width=1.3mm, red!60] (-0.3,-0.95) -- (1.center)
					-- (-0.4, 0.5);
\draw[line width=1.3mm, blue!60] (2.45,-0.05) -- (2.center);

\draw[line width=1.3mm, blue!60] 
				([xshift=-1.2mm, yshift=-1.2mm]3.center) 
					-- (-0.32,0.13);
\draw[line width=1.3mm, blue!60] 
				([xshift=-1.5mm, yshift=0.5mm]5.center) 
					-- (-0.35,-0.575);
					
\draw[line width=1.3mm, blue!60] 
				(8.center) -- (1.9,1);					
\draw[line width=1.3mm, red!60] 
				([xshift=0.2mm, yshift=-1.7mm]6.center) 
				-- ([xshift=3.2mm, yshift=-1mm]2.center) 
				-- ([xshift=1.5mm, yshift=1mm]7.center);

\draw[thick]  (1.center) -- (3.center);
\draw[thick]  (1.center) -- (4.center);
\draw[thick]  (1.center) -- (5.center);
\draw[thick]  (6.center) -- (2.center);
\draw[thick]  (7.center) -- (2.center);
\draw[thick]  (8.center) -- (2.center);

\foreach \i in {1,...,8}{
	\draw[thick, fill=black] (\i.center) circle (0.06);
}
\end{tikzpicture}
\hspace{0.5cm}
\begin{tikzpicture}[scale=0.65]
\node at (0,0) (1) {};
\node at (2,0.5) (2) {};
\node at (-0.8,1) (3) {};
\node at (-1.2,-0.7) (4) {};
\node at (-0.6,-1.9) (5) {};
\node at (3.2,0.9) (6) {};
\node at (2.9,-0.6) (7) {};
\node at (1.8,1.5) (8) {};

\node at ([xshift=12mm, yshift=-3mm]1.center) {$B_u=B_v$};

\draw[line width=1.3mm, red!60] (-0.3,-0.95) -- (1.center)
					-- (2.center) -- (2.6, 0.7);
\draw[line width=1.3mm, red!60] (4.center) -- (1.center)
					-- (-0.4, 0.5);
\draw[line width=1.3mm, red!60] (2.45,-0.05) -- (2.center);

\draw[line width=1.3mm, blue!60] (0.5, 0.305) -- (1.9, 0.655);

\draw[line width=1.3mm, blue!60] 
				([xshift=-1.2mm, yshift=-1.2mm]3.center) 
					-- (-0.32,0.13);
\draw[line width=1.3mm, blue!60] 
				([xshift=-1.5mm, yshift=0.5mm]5.center) 
					-- (-0.35,-0.575);
					
\draw[line width=1.3mm, red!60] 
				(8.center) -- (1.9,1);					
\draw[line width=1.3mm, blue!60] 
				([xshift=0.2mm, yshift=-1.7mm]6.center) 
				-- ([xshift=3.2mm, yshift=-1mm]2.center) 
				-- ([xshift=1.5mm, yshift=1mm]7.center);

\draw[thick]  (1.center) -- (2.center);
\draw[thick]  (1.center) -- (3.center);
\draw[thick]  (1.center) -- (4.center);
\draw[thick]  (1.center) -- (5.center);
\draw[thick]  (6.center) -- (2.center);
\draw[thick]  (7.center) -- (2.center);
\draw[thick]  (8.center) -- (2.center);

\foreach \i in {1,...,8}{
	\draw[thick, fill=black] (\i.center) circle (0.06);
}
\end{tikzpicture}
\caption{On the left, we have the two different initial cases, i.e.,
on the top,~$B_u\neq B_v$, on the bottom,~$B_u = B_v$. In
the middle, the recursive call is made. On the right, we use the
two recursive colorings and swap colors if needed.
}
\label{fig:tree-recursive-NM}
\end{center}
\end{figure}
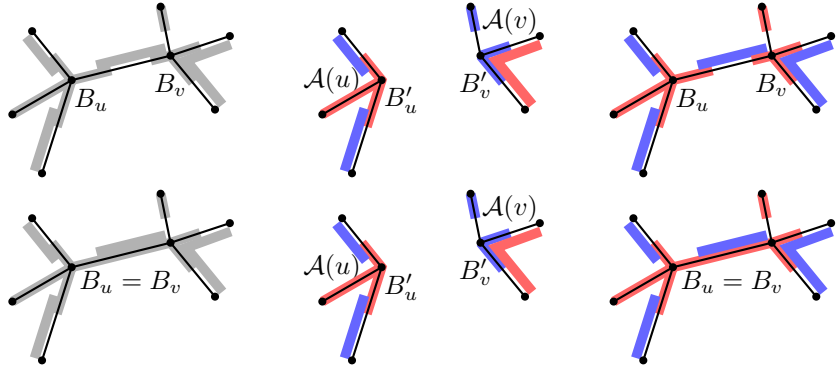
\begin{theorem}\label{thm:NM-balls-on-trees}
$ \NMCN{balls}{trees} (t; n)=2$.
\end{theorem}
\begin{proof}
The coloring obviously uses two colors. It remains to show
it is non-monochromatic. We use induction on~$t$. If~$t=0$,
the coloring is non-monochromatic since it uses the chain method.

Suppose now that~$t\geqslant 1$ and that the claim holds for any
tree space with fewer than~$t$ internal nodes.
Let~$p$ be a point contained in at
least two balls.

If~$p$ is contained in balls
only of~$\objects(v)$, only of~$\objects(u)$,
or only of~$\objects(e)$, it is
contained in at least two balls of different colors. Indeed, the
colorings of~$\objects(v)$ and~$\objects(u)$ are non-monochromatic
since they use the method on a tree with fewer than~$t$ internal
nodes and we can use the induction hypothesis. Moreover~$\objects(e)$
is non-monochromatic due to the chain method.

It remains to consider the case where~$p$ is contained in balls from at least two of the
sets~$\objects(u)$, $\objects(v)$, and~$\objects(e)$.
We distinguish two cases:~$p$ is contained in a ball
of~$\objects(e)$ and~$p$ is not contained in a ball
of~$\objects(e)$.

If~$p$ is contained in a ball~$B$ of~$\objects(e)$,
we can assume without loss
of generality that~$p$ is also contained in a ball
of~$\objects(v)$. By Lemma \ref{lem:connected_core}(i),  we 
have that~$p \in B_v$.

If~$B_u = B_v$
then all balls in~$\objects(e)$ are given a different color
than~$B_v$ hence~$p$ is contained in two balls of different color.
If~$B_u \neq B_v$ then we use the chain method on~$\pi$. Hence
if~$p\in \pi$, it is contained in two balls of different color.
To show that if~$p\notin \pi$ then~$p$ is still contained in two balls
of different colors,
it suffices to notice that for any subset of balls of~$\objects(e)$
in which~$p$ is contained,
the point~$p' \in \pi$ at distance~$d(u,p)$ from~$u$ is contained 
in the same set of
balls from~$\objects(e)$ as~$\pi$ is the longest path containing~$e$. 

On the other hand, if~$p$ is not contained in a ball
of~$\objects(e)$, then it is contained in at least one ball
from~$\objects(u)$ and one from~$\objects(v)$. By 
Lemma~\ref{lem:connected_core} we have that~$p \in B_u \cap B_v$.

We then have two cases. If~$B_u = B_v$, then~$p$ is contained in
another ball of~$\objects(u)$ or~$\objects(v)$, and then the
coloring is non-monochromatic by the induction hypothesis.
Otherwise~$B_u$ and~$B_v$ are part of the
chain~$\objects(e) \cup \{ B_u, B_v \}$, and hence~$p$ is contained
in at least two balls of different color.
  \end{proof}

\mypara{CF-coloring.}
The second algorithm CF-colors~$\objects$
using~$\lceil \log t \rceil + 3$ colors.
As before, define~$\core:=\{B=B(c,r) \mid \exists x: B=B_x \}$.
We explain how to color~$\core$ and then extend the
coloring to~$\objects$.
Let~$r$ be a node whose removal results in
subtrees each of at most~$t/2$ internal nodes.
We color~$B_r$ (if it exists) with color~1.
Let~$\treespace_1,\ldots,\treespace_{{\mathrm deg}(r)}$ be
subtrees resulting from removing~$r$, that is,
the closures of the connected components of $\treespace\setminus\{r\}$.
For each~$i=1,\ldots, {\mathrm deg}(r)$, we recurse
on~$\treespace_i$ with the balls from $\core$
whose centers lie in~$\treespace_i$. In such a recursive call,
we consider a node to be an internal node when it was an internal
node in the original space~$\treespace$ and when it has not yet
been selected as a splitting node in a previous call.  Hence,
when $t=0$ in a recursive call on a subtree $\tree'\subset \treespace$,
then $\treespace'$ must be a single edge both of whose endpoints have
already been treated.

The recursion stops
when there are no more balls left (which must be the case
when we have a recursive call with $t=0$).
Note that the internal nodes are fixed from the beginning, hence
at some point of the recursion, a leaf node might still be
considered internal for the purposes of the recursion.

\begin{lemma}\label{lem:balls-on-trees-CF-core}
The above algorithm CF-colors~$\core$ using~$\lceil \log t \rceil$
colors.
\end{lemma}
\begin{proof}
The number of colors used comes immediately from the splitting
of~$\treespace$ into trees of at most~$\frac{t}{2}$ internal
nodes. We now show the coloring is indeed confict-free by
showing that it is a unimin coloring: for any point~$p$
the minimum color among the colors of the balls containing~$p$
is unique.
Let~$p\in\treespace$ be a point contained in two balls~$B_1=B(p_1,r_1)$
and~$B_2=B(p_2,r_2)$ both of color~$i$. We show that this implies
the existence of a ball of higher color containing~$p$.
Let~$v_1$ be the node~$B_1$
is assigned to, and ~$v_2$ the node~$B_2$ is assigned to.
Since~$B_1$ and~$B_2$ have the color~$i$, they were contained in
different trees when they were colored in the recursive process.
Let~$v_0$ be the node that disconnected~$v_1$ and~$v_2$ and
let~$B_0$ be the ball assigned to~$v_0$. Note that~$c(B_0)<i$.

We prove that~$p\in B_0$. Let~$\pi$ be the unique simple path
between~$p$ and~$v_0$. It cannot be the case that both~$p_1 \in \pi$
and~$p_2 \in \pi$. Suppose without loss of generality
that~$p_2 \notin \pi$. Let~$d$ be the distance between~$p$ and~$v_0$.
Since~$p\in B_2$, we have that~$\cov_{v_0}(B_2) \geqslant d$.
And since~$\cov_{v_0}(B_0) \geqslant \cov_{v_0}(B_2)$, we
have that~$p\in B_0$, concluding the proof.
  \end{proof}

We now wish to extend the coloring to balls in $\objects\setminus \core$.
To this end,
define~$\treespace' := \treespace \setminus (\bigcup \core)$ to be
the part of $\treespace$ that remains after removing
all points covered by the balls in~$\core$.

We finish the coloring with three more colors (using the chain method
for CF-colorings) as explained next, resulting in~$\lceil \log t \rceil +3$ colors.
We use the following lemma to show that the remaining
balls can be reduced to intervals on disjoint lines.
Note that it does not use tree spaces and can hence be applied
also for planar network spaces.

\begin{lemma}\label{lem:not-core-then-in-single-edge}
For any ball~$B\notin \core$, we have
$\{p \in B \mid p\notin \cup \core \} \subseteq e, $
where~$e$ is the edge containing the center of~$B$.
\end{lemma}
\begin{proof}
Suppose for a contradiction that there is a point~$p\notin e$
contained in~$B$ but not in~$\cup \core$. Consider the endpoint~$v$
of~$e$ belonging to the geodesic from the center of~$B$ to~$p$.
We claim that~$\cov_v(B)>\cov_v(B_v)$, contradicting the definition of~$B_v$. 
Indeed, $\cov_v(B)>d(v,p)$  (since $v$ lies on the geodesic from $B$'s center
to~$p$) and $\cov_v(B_v)<d(v,p)$ (since $p\not\in \core$ and,
hence, $p\notin B_v$).
  \end{proof}

\begin{theorem}
$\CFCN{tree}{balls} (t; n) \leqslant \lceil \log t \rceil +3$.
\end{theorem}

\subsection{Tree spaces: the lower bound}
\begin{lemma}
$\CFCN{tree}{balls} (t; n) \geqslant \left\lceil\log (t+1) \right\rceil.$
\end{lemma}
\begin{proof}
Let~$\treespace$ be as follows. We take~$t+2$ points~$p_1,\ldots,p_{t+2}$
in the plane, with~$p_i=(i, 0)$ for each~$i=1,\ldots, t+2$, and we
link consecutive points with a unit distance segment. We then
take~$t+2$ additional points~$p'_1,\ldots, p'_{t+2}$,
with~$p'_i=(i,t+2)$, and for each~$i=1,\ldots,t+2$ we link~$p_i$
and~$p'_i$ with a segment of length~$t+2$. Note that~$p_1$
and~$p_{t+2}$ do not count as internal nodes as their degree is two.
Finally, we place~$t+1$
balls~$B_1=B(c_1,t+2),\ldots, B_{t+1}=B(c_{t+1}, t+2)$,
 for all~$i=1,\ldots,t+1$, with~$c_i=(i+\frac{2}{3}, 0)$,
see Fig.~\ref{fig:comb}.
\begin{figure}
\begin{center}
\begin{tikzpicture}[scale=0.65]
\draw[line width=1.3mm, red!50, shorten >=-0.4mm]
		(1,2.34) -- (1,0) -- (5,0) -- (5,1.66);
\draw[line width=1.3mm, red!50, shorten >=-0.4mm]
		(2,0) -- (2,3.34);
\draw[line width=1.3mm, red!50, shorten >=-0.4mm]
		(3,0) -- (3,3.66);
\draw[line width=1.3mm, red!50, shorten >=-0.4mm]
		(4,0) -- (4,2.66);
\node at (2.5, 3) {$B_2$};

\foreach \i in {1,2,...,5}{
	\node at (\i,0) (\i) {};
	\draw[thick, fill=black] (\i,0) circle (0.06);
	\draw[thick, fill=black] (\i,4) circle (0.06);
	\draw[thick] (\i,0) -- (\i,4);
	\node at ([yshift=-3mm]\i.center) {$p_\i$};
	\node at ([yshift=43mm]\i.center) {$p'_\i$};
}

\draw[thick] (1.center) -- (2.center) -- (3.center)
			 -- (4.center) -- (5.center);

\foreach \i in {1,2,...,4}{
	\draw[thick, fill=white]
			([xshift=6.6mm]\i.center) circle (0.06);
	\node at ([xshift=6.6mm, yshift=3mm]\i.center) {$c_\i$};
}
\end{tikzpicture}
\end{center}
\caption{Example of the lower bound construction with~$t=3$.
For clarity purposes, only~$B_2$ is displayed, in red. }
\label{fig:comb}
\end{figure}
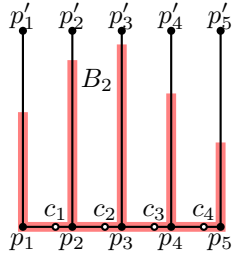
Consider the hypergraph~$\mathcal{H}$ whose nodes are the balls~$B_i$,
and whose hyperedges are the subsets of balls such that there is a 
point~$p\in \treespace$ contained in exactly that subset
(and no other balls). We claim (and will prove below) that the set
of hyperedges is exactly the
set~$\{ \{ B_i,B_{i+1},\ldots, B_j \} \mid i \geqslant j \}$.
In other words, there is a hyperedge for a subset of balls if and
only if there is an interval on the $x$-axis containing exactly
the centers of these balls. Hence, we can apply
the~$\ceil{\log (t+1)}$ lower bound for CF-coloring points
with respect to intervals~\cite{even-cf-03}. 

To prove the claim, note that if~$p_i$ is the ball center
nearest to~$p$ then~$d(p,p_1)> d(p,p_2)> \cdots > d(p,p_i)$
and~$d(p,p_{i+1})> \cdots > d(p,p_{t+2})$, which implies that
any hyperedge is of the form $\{B_i,B_{i+1},\ldots,B_j\}$.
On the other hand, the point~$(\floor{(j+i)/2},t+2-(j-i)/2)$ 
is contained in
exactly the balls~$B_i,B_{i+1},\ldots,B_j$.
  \end{proof}

\subsection{Planar network spaces}\label{sec:balls-on-networks}

\mypara{NM-coloring.}
We first explain how to NM-color balls on a planar network space~$\graph$.
Let again~$\core$ be the set~$\{B=B(c,r) \mid \exists x: B=B_x \}$.
We create a graph~$\mathcal{G}_\core$
whose node set is~$\core$ and whose edge set is defined as
follows: there is an edge between~$B$ and~$B'$ if and only if
there is an edge $vv'$ in~$\treespace$
with~$B_v=B$ and~$B_{v'}=B'$.
It follows from Lemma~\ref{lem:connected_core} that for any
ball~$B$, the set of nodes of~$\graph$ to which~$B$ is assigned,
together with the edges between these nodes,
is a connected set.
Therefore,~$\mathcal{G}_\core$ is planar as well since
its nodes correspond to disjoint connected subspaces in the planar
space~$\graph$.
We now use the Four Color Theorem to color~$\mathcal{G}_\core$ and
we give each ball in~$\core$ the same color as the corresponding
node in~$\mathcal{G}_\core$.

\begin{lemma}\label{lem:net-core}
The coloring on~$\core$ is non-monochromatic and uses at most four colors.
\end{lemma}
\begin{proof}
It is clear that the coloring uses at most four colors. Now let~$p$
be a point contained in two balls~$B_1$ and~$B_2$ of the same color.
Let~$v_1$ and~$v_2$ be nodes of~$\graph$ with~$B_1=B_{v_1}$
and~$B_2=B_{v_2}$. Let~$\pi_1$ and~$\pi_2$
be two shortest paths between~$p$ and~$v_1, v_2$, respectively.
If all the nodes in~$\pi_1 \cup \pi_2$ are either assigned~$B_1$
or~$B_2$, then there is an edge between~$B_1$ and~$B_2$
in~$\mathcal{G}_\core$
and hence~$B_1$ and~$B_2$ are given different colors. Therefore
there must be a node~$v$ in~$\pi_1 \cup \pi_2$ (we assume
without loss of generality that~$v\in \pi_1$)
with~$B_v\notin \{ B_1, B_2 \}$ and~$c(B_v)=c(B_1)$. Note that
if~$c(B_v)=c(B_1)$ for all~$v\in \pi_1$, then there must be an edge
between two balls of the same color in~$\mathcal{G}_\core$ which
is a contradiction, hence there must be a vertex~$v\in \pi_1$
with~$c(B_v)\neq c(B_1)$. Since~$\pi_1$ is a shortest path
between~$v_1$ and~$p$, and since~$v\in \pi_1$, we have that~$\pi_1$
contains a shortest path between~$v$ and~$p$.
Moreover,~$\cov_v(B_v) \geqslant \cov_v(B_1) \geqslant d(v,p)$,
which implies that~$p \in B_v$ and
concludes the proof.
  \end{proof}

We now wish to extend the coloring to balls in $\objects\setminus \core$.
To this end, define $\graph' := \graph \setminus (\bigcup \core)$ to be
the part of $\graph$ that remains after removing all points covered by the
balls in~$\core$. The proof of the following lemma is similar to the proof
of Lemma~\ref{lem:not-core-then-in-single-edge}.

\begin{lemma}
Consider a ball $B\in \objects\setminus\core$, and
let~$B' := B \cap \graph'$. Then $B'$ is contained in a single
edge of $\graph'$.
\end{lemma}

For each edge $e$ of $\graph'$, let $\objects(e)$ denote the set of balls contained in~$e$.
Let $u$ and $v$ denote the endpoints of the edge in $\graph$ containing~$e$.
We color the uncolored balls in~$e$ using
the chain method with two colors not equal to~$c(B_u)$ and~$c(B_v)$.
We have now colored the balls in $\core$ as well as the balls in
$\objects\setminus\core$ that lie at least partially in $\graph'$.
Next we explain how to color the remaining balls, which are fully covered by the balls in $\core$.

\begin{lemma}\label{lem:3-balls}
Any uncolored ball is contained in the union of at most three balls.
\end{lemma}
\begin{proof}
Any uncolored ball~$B$ is contained in~$\cup \core$.
If~$B$ is fully contained in a single edge~$e$ of~$\graph$,
it must be covered by the two balls
from~$\core$ extending the farthest into~$e$, starting from
each of the two endpoints. If not, let~$v$ be a node contained in~$B$.
Now~$B \setminus B_v$ is contained in a single edge~$e$ of~$\graph$
and so $B\setminus B_v$ can be covered by two balls (as just explained), 
which implies that $B$ can be covered by three balls.
  \end{proof}

Using this lemma, we can easily finish the NM-coloring.

\begin{theorem}
$\NMCN{planar}{balls} (t; n) = 4$.
\end{theorem}
\begin{proof}
The coloring obviously uses four colors at most. Moreover, it is
easy to see the coloring is non-monochromatic. It remains to
show that there is an instance requiring at least four colors.
To that purpose, let~$\graph$ be an embedding of~$K_4$ where all edges
have length one. Then, for each node~$v$ of~$\graph$, we create
the ball~$B(v,2/3)$. Since no two balls can have the same color,
we need at least four colors.
  \end{proof}

\mypara{CF-coloring.}
We now explain how to CF-color balls on a planar network.
As before, define~$\core := \{B=B(c,r) \mid \exists x: B=B_x \}$. 
We first CF-color $\core$ using 
the following recursive algorithm introduced by Smorodinsky~\cite{smor-geomCF-06}:
we select a maximum independent set in~$C_1:=\core$, we give
it color 1, place all uncolored balls in~$C_2$, and recurse. We claim
that for all~$i$, the Delauney graph~$D_i:=(C_i,E_i)$ on the balls in~$C_i$
is planar, where~$E_i:= \{ \{B_1, B_2\} \mid \exists p\in \graph: p \in B_1 \cap B_2 \text{ and } \forall B\notin \{ B_1, B_2\} :  p \notin B \} $.

\begin{lemma}\label{lem:D_i-planar}
$D_i$ is planar.
\end{lemma}
\begin{proof}
We draw~$D_i$ using the drawing of~$\graph$ as follows:
each ball is represented by its center. Then, for every
edge in~$D_i$, we find a witness, that is a point
contained in the intersection of the two balls and not in
any other ball. We finally draw the edge as two geodesics
on~$\graph$: one from
one endpoint to the witness point, and the other from the witness
point to the other endpoint.

We claim that this drawing is plane.
Suppose by contradiction that it is not the case and there
is a crossing between the two edges~$B_1B_3$ 
and~$B_2B_4$. Suppose also that the
endpoints of the two edges are distinct: the argument
when an endpoint is shared is similar. Since we based
our drawing on~$\graph$, a planar graph, the point
where the two edges cross must be a node~$x$ in~$\graph$. 
Let~$w_{13}$ be the witness of the edge~$B_1B_3$ and~$w_{24}$ 
the winess of~$B_2B_4$. 
Fig.~\ref{fig:crossing} shows the two crossing edges,
with the crossing node~$x$ in the middle, and the
two witnesses~$w_{13}$ and~$w_{24}$ used to draw the geodesics.

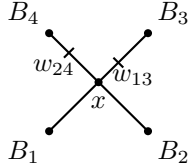
\begin{figure}
\begin{center}
\begin{tikzpicture}[scale=0.65]
\node at (-1,-1) (a) {};
\node at (1,-1) (b) {};
\node at (0,0) (c) {};
\node at (-1,1) (d) {};
\node at (1,1) (e) {};

\foreach \a/\l/\loc in {a/B_1/below left,
					 b/B_2/below right,
					 c/ /below,
					 d/B_4/above left,
					 e/B_3/above right}{
	\draw[thick, fill] (\a) circle (0.06);
	\node[\loc] at (\a) {$\l$};
}
\node at ([yshift=-4mm]c.center) {$x$};

\draw[thick] (a.center) -- (c.center) -- (d.center);
\draw[thick] (b.center) -- (c.center) -- (e.center);
\draw[thick] (-0.7,0.5) --++ (0.2,0.2);
\node at (-0.9,0.3) () {$w_{24}$};
\draw[thick] (0.5,0.3) --++ (-0.2,0.2);
\node at (0.7,0.1) () {$w_{13}$};
\end{tikzpicture}
\end{center}
\caption{We suppose for a contradiction that the edges~$B_1B_3$ 
and~$B_2B_4$ cross. The crossing point is a node~$x$ of~$\graph$. 
Let~$w_{13}$ be the witness of the edge~$B_1B_3$ and~$w_{24}$ 
the winess of~$B_2B_4$. }
\label{fig:crossing}
\end{figure}

Suppose, without loss of generality,
that the distance from~$x$ to~$w_{24}$ is greater than or equal to 
the distance from ~$x$ to~$w_{13}$. Thus, the distance from the center
of~$B_1$ to~$w_{24}$ is greater than or equal to~$w_{13}$. 
Hence,~$w_{13}$ is also contained
in the ball~$B_1$, which contradicts the definition of a
witness. Thus, the drawing is plane.
  \end{proof}

Using this lemma and the Four Color Theorem, we get a coloring
on~$\core$ using~$\lceil \log_{4/3} t \rceil$ colors. 
Note that this method does not give an efficient algorithm
because of the use of the Four Color Theorem. For a fast algorithm, we can
use a linear-time 
algorithm~\cite{Chiba_5_coloring} to find an independent set of
size at least~$n/5$, leading
to~$\lceil \log_{5/4} t \rceil$ colors. 

We then color the balls in~$\objects \setminus \core$.
Using Lemma~\ref{lem:not-core-then-in-single-edge}, we have
that for any such ball~$B$, the set of points contained
in~$B$ but not in any ball in~$\core$ is contained
in one edge of~$\graph$.
Therefore, if we cut~$\cup\mathcal{C}$ out of~$\graph$,
the remaining space is a union of disjoint segments,
and any object that is not colored is contained in
at most one segment. We can therefore use the chain coloring on
each segment with the two additional colors and the dummy one.

Finally, any point in~$\cup\core$ is contained
in a ball in~$\core$ of unique color, and any point
not in~$\cup\core$, is contained in at most one ball
of each of the two additional colors. Therefore,
the coloring is conflict-free. This yields the following theorem.

\begin{theorem}
$\CFCN{planar}{balls} (t; n) \leqslant \lceil \log_{4/3} t \rceil +3$.
\end{theorem}

\section{Concluding Remarks} 
We studied NM- and CF-colorings on network spaces, where the
objects to be colored are connected regions of the network space. 
We showed that the
number of colors can be bounded as a function of the complexity 
(which depends on the type of space and of objects) 
of the network space 
and the objects, rather than on the number of objects. 
All our bounds are tight up to some constants, except 
for $\CFCN{tree}{trees} (k,\ell; n)$
where the ~upper bound 
is a factor~$\ell$ away from the lower bound. Closing this gap 
remains an open problem. It would also be interesting to find 
bounds on general connected objects on any network space, or 
other settings where the number of colors depends on the complexity 
of the space and objects rather the number of objects. 

\bibliography{creatureCFbib}
\newcommand{\cgta}{\emph{Comput.\ Geom.\ Theory Appl.}\xspace}
\newcommand{\dcg}{\emph{Discr.\ Comput.\ Geom.}\xspace}
\newcommand{\ijcga}{\emph{Int.\ J.\ Comput.\ Geom.\ Appl.}\xspace}
\newcommand{\algor}{\emph{Algorithmica}\xspace}
\newcommand{\sicomp}{\emph{SIAM J.\ Comput.}\xspace}
\newcommand{\jalg}{\emph{J.\ Alg.}\xspace}
\newcommand{\jda}{\emph{J.\ Discr.\ Alg.}\xspace}
\newcommand{\ipl}{\emph{Inf.\ Proc.\ Lett.}\xspace}
\newcommand{\jacm}{\emph{J.~ACM}\xspace}

\newcommand{\socg}[1]{In \emph{Proc.\ #1 ACM Symp.\ Comput.\ Geom.}\xspace}
\newcommand{\soda}[1]{In \emph{Proc.\ #1 ACM-SIAM Symp.\ Discr.\ Alg.}\xspace}
\newcommand{\stoc}[1]{In \emph{Proc.\ #1 ACM Symp.\ Theory Comp.}\xspace}
\newcommand{\esa}[1]{In \emph{Proc.\ #1 Europ.\ Symp.\ Alg.}\xspace}
\newcommand{\swat}[1]{In \emph{Proc.\ #1 Scandinavian Workshop.\ Alg.\ Theory}\xspace}

\end{document}